\documentclass[lettersize,journal]{IEEEtran}
\usepackage{amsmath,amsfonts,amssymb,bm,amsthm}
\usepackage{color}
\usepackage{algorithmic}
\usepackage{array}
\usepackage{caption}
\usepackage[justification=raggedright,singlelinecheck=false,font=footnotesize]{caption}
\usepackage[subrefformat=parens,labelformat=simple,justification=centerlast]{subcaption}
\usepackage{textcomp}
\usepackage{algorithm, algorithmic}
\usepackage{stfloats}
\usepackage{url}
\usepackage{verbatim}
\usepackage{graphicx}
\def\BibTeX{{\rm B\kern-.05em{\sc i\kern-.025em b}\kern-.08em
		T\kern-.1667em\lower.7ex\hbox{E}\kern-.125emX}}
\usepackage{balance}
\usepackage{threeparttable}
\usepackage{pifont}
\usepackage{bbding}
\usepackage{multirow}
\usepackage{makecell}
\usepackage{booktabs}
\usepackage{epstopdf}
\usepackage{epsfig, latexsym, times}
\usepackage{float}
\usepackage{subcaption}
\usepackage{cases}
\usepackage{stfloats}
\usepackage[colorlinks=true, allcolors=blue]{hyperref}

\newtheorem{proposition}{Proposition}

\usepackage{acro}
\DeclareAcronym{isac}{
	short = ISAC,
	long  = Integrated Sensing and Communication
}

\DeclareAcronym{dfrc}{
	short = DFRC,
	long  = Dual Functional Radar and Communication
}
\DeclareAcronym{ra}{
	short = RA,
	long  = Reconfigurable Antenna
}
\DeclareAcronym{rf}{
	short = RF,
	long  = Radio Frequency
}

\DeclareAcronym{em}{
	short = EM,
	long  = Electromagnetic
}

\DeclareAcronym{fd}{
	short = FD,
	long  = Fully Digital
}
\DeclareAcronym{mu}{
	short = MU,
	long  = Multi-user
}

\DeclareAcronym{su}{
	short = SU,
	long  = Single-user
}

\DeclareAcronym{mimo}{
	short = MIMO,
	long  = Multiple-input Multiple-output
}

\DeclareAcronym{miso}{
	short = MISO,
	long  = Multiple-input Single-output
}

\DeclareAcronym{sust}{
	short = SUST,
	long  = Single-user Single-target
}

\DeclareAcronym{mumt}{
	short = MUMT,
	long  = Multiple-user Multiple-target
}

\DeclareAcronym{cus}{
	short = CUs,
	long  = Communication Users
}
\DeclareAcronym{cu}{
	short = CU,
	long  = Communication User,
	first-style = short
}
\DeclareAcronym{hbf}{
	short = HBF,
	long  = Hybrid Beamforming
}
\DeclareAcronym{dof}{
	short = DoF,
	long  = Degree of Freedom
}
\DeclareAcronym{dofs}{
	short = DoFs,
	long  = Degrees of Freedom
}
\DeclareAcronym{era}{
	short = ERA,
	long  = Electromagnetically Reconfigurable Antenna
}
\DeclareAcronym{oa}{
	short = OA,
	long  = Omnidirectional Antenna
}

\DeclareAcronym{ba}{
	short = BA,
	long  = Broadside Antenna
}

\DeclareAcronym{cgv}{
	short = CGV,
	long  = Complex Channel Gain Vector
}
\DeclareAcronym{agv}{
	short = AGV,
	long  = Antenna Gain Vector
}
\DeclareAcronym{arv}{
	short = ARV,
	long  = Antenna Response Vector
}
\DeclareAcronym{ue}{
	short = UE,
	long  = User Equipment
}

\DeclareAcronym{ula}{
	short = ULA,
	long  = Uniform Linear Array
}

\DeclareAcronym{upa}{
	short = UPA,
	long  = Uniform Planar Array
}

\DeclareAcronym{los}{
	short = LoS,
	long  = Line-of-Sight
}
\DeclareAcronym{aod}{
	short = AoD,
	long  = Angle-of-Departure
}

\DeclareAcronym{awgn}{
	short = AWGN,
	long  = Additive White Gaussian Noise
}
\DeclareAcronym{snr}{
	short = SNR,
	long  = Signal-to-Noise Ratio
}
\DeclareAcronym{sinr}{
	short = SINR,
	long  = Signal-to-Interference-plus-Noise Ratio
}
\DeclareAcronym{scnr}{
	short = SCNR,
	long  = Signal-to-Clutter-plus-Noise Ratio
}
\DeclareAcronym{mse}{
	short = MSE,
	long  = Mean Squared Error
}
\DeclareAcronym{bcd}{
	short = BCD,
	long  = Block Coordinate Descent
}

\DeclareAcronym{ao}{
	short = AO,
	long  = Alternating Optimization
}

\DeclareAcronym{fp}{
	short = FP,
	long  = Fractional Programming
}

\DeclareAcronym{mo}{
	short = MO,
	long  = Manifold Optimization
}

\DeclareAcronym{sip}{
	short = SIP,
	long  = Semi-Infinite Programming
}

\DeclareAcronym{rmi}{
	short = RMI,
	long  = Radar Mutual Information
}

\DeclareAcronym{sandc}{
	short = S\&C,
	long  = Sensing and Communication
}

\DeclareAcronym{itu}{
	short = ITU,
	long  = International Telecommunication Union
}

\DeclareAcronym{thbf}{
	short = Tri-HBF,
	long  = Triple-Hybrid Beamforming
}

\DeclareAcronym{sdr}{
	short = SDR,
	long  = Semidefinite Relaxation
}
\DeclareAcronym{sdp}{
	short = SDP,
	long  = Semidefinite Programming
}

\DeclareAcronym{kkt}{
	short = KKT,
	long  = Karush-Kuhn-Tucker 
}

\DeclareAcronym{wrt}{
	short = w.r.t.,
	long  = with respect to
}

\DeclareAcronym{doa}{
	short = DoA,
	long  = Direction-of-Arrival
}

\DeclareAcronym{espar}{
	short = ESPAR,
	long  = Electronically Steerable Parasitic Array Radiator 
}

\DeclareAcronym{dma}{
	short = DMA,
	long  = Dynamic Metasurface Antenna
}
\DeclareAcronym{crb}{
	short = CRB,
	long  = Cramér–Rao Bound
}
\DeclareAcronym{cpu}{
	short = CPU,
	long  = Central Processing Unit
}

\makeatletter
\renewenvironment{proof}[1][\proofname]{\par
	\pushQED{\qed}%
	\normalfont \topsep6\p@\@plus6\p@\relax
	\trivlist
	\item[\hskip\labelsep
	\itshape 
	#1\@addpunct{:}]\ignorespaces  
}{
	\popQED\endtrivlist\@endpefalse
}
\makeatother

\renewcommand{\proofname}{Proof}

\begin{document}
	\title{Tri-Hybrid Beamforming Design for ISAC Systems with  Reconfigurable Antennas}
	\author{Jiangong Chen, \IEEEmembership{Graduate Student Member,~IEEE}, 
		Xia Lei,
		Yuchen Zhang, \IEEEmembership{Member,~IEEE}, 
		Kaitao Meng, \IEEEmembership{Member,~IEEE}, and
		Christos Masouros, \IEEEmembership{Fellow,~IEEE}
		\vspace{-10mm}
		
		\thanks{

			J. Chen and X. Lei are with the National Key Laboratory of Wireless Communications, University of Electronic Science and Technology of China, Chengdu, China (e-mail: jg\_chen@std.uestc.edu.cn,  leixia@uestc.edu.cn).
			
			Y. Zhang is with the Electrical and Computer Engineering Program, Computer, Electrical and Mathematical Sciences and Engineering (CEMSE), King Abdullah University of Science and Technology (KAUST), Thuwal 23955-6900, Kingdom of Saudi Arabia (e-mail: yuchen.zhang@kaust.edu.sa).
					
			K. Meng is with the Department of Electrical and Electronic Engineering, University of Manchester, Manchester, UK (email: kaitao.meng@manchester.ac.uk).
					
			C. Masouros is with the Department of Electronic and Electrical Engineering, University College London, London, UK (e-mail: c.masouros@ucl.ac.uk).
			
		}
	}
	\markboth{Journal of \LaTeX\ Class Files,~Vol.~18, No.~9, September~2020}%
	{How to Use the IEEEtran \LaTeX \ Templates}
	
	\maketitle\textbf{}
	\maketitle
	
	\begin{abstract}
		\ac{isac} systems require efficient beamforming architectures to jointly support communication and sensing functionalities. To reduce hardware overhead, \ac{hbf} has been widely studied and shown to achieve performance close to fully digital beamforming under practical hardware constraints. As a promising evolution, \ac{ra} technologies have recently emerged to further enhance beamforming \ac{dofs} by dynamically reconfiguring antenna \ac{em} characteristics, yet their integration into \ac{isac} systems remains largely unexplored. In this paper, we investigate an \ac{ra}-assisted \ac{isac} system and develop a decoupled \ac{thbf} framework that alternatively optimizes digital, analog, and \ac{em} beamformers to maximize the communication rate and sensing \ac{scnr}. 
		For both \ac{sust} and \ac{mumt} scenarios, we first transform the original fractional objectives into fraction-free ones via methods tailored to their respective structures. The resulting problems are then solved via alternating optimization over different variable blocks. Closed-form updates are derived for all variables except the \ac{em} beamforming subproblem in the \ac{mumt} scenario. To further reduce complexity introduced by \ac{sdr} in \ac{em} beamforming, we propose a low-complexity iterative approach across antennas with closed-form updates.
		Simulation results demonstrate that the proposed scheme significantly outperforms benchmark designs with conventional omnidirectional and directional antennas, achieving almost 100\% improvement in spectrum efficiency and 62.5\% reduction in antenna overhead, thereby unveiling the potential of \ac{ra} in enhancing \ac{isac} performance towards 6G.

	\end{abstract}
	\acresetall

	
	\section{Introduction}
	As one of the six usage scenarios of IMT-2030 \cite{IMT2030}, \ac{isac} has become a cornerstone technology for the forthcoming 6G networks, offering significant advantages in efficiently combining communication and sensing functionalities \cite{ISACBack1,ISACBack2}. By integrating both capabilities, \ac{isac} systems make more efficient use of hardware and spectrum resources, significantly enhancing the sustainability of wireless networks.

	A key enabler of \ac{isac} is \ac{mimo} technology, which offers abundant \acp{dof} for both communication and sensing. On the communication side, \ac{mimo} techniques such as beamforming, space-time coding, and \ac{mu} precoding provide high spectral efficiency, spatial diversity, and interference suppression. On the sensing side, \ac{mimo} processing also enables high-resolution \ac{doa} estimation, active sensing, and robust target detection. Motivated by these advantages, a large body of literature has investigated beamforming design for \ac{isac} systems. Representative formulations directly optimize communication and sensing \ac{sinr}/\ac{scnr} \cite{isacbfsinr1,isactanbo,dp}, adopt transmit beampattern matching \cite{isacbf1,isacbf2,isacbf3}, exploit estimation-theoretic criteria such as the \ac{crb} \cite{isaccrb}, or employ sensing mutual information as an information-theoretic metric \cite{distortion,isacmsi1,isacmsi2}. These studies have established the fundamental algorithmic framework of \ac{isac} beamforming, but most of them are mainly developed under conventional fully digital array architectures.

	As wireless systems continue to migrate toward mmWave and higher frequency bands, large-scale arrays become increasingly attractive for \ac{isac}, since narrow beams and large apertures are beneficial to both communication and sensing. However, fully digital beamforming requires one dedicated \ac{rf} chain per antenna, leading to prohibitive hardware cost and power consumption in such regimes. To address this issue, \ac{hbf} has been widely recognized as a practical solution \cite{HBF}. Existing \ac{isac} works have investigated fully-connected \ac{hbf} architectures \cite{isachbf1}, sub-connected architectures \cite{isachbf2}, and improved sub-connected implementations with enhanced phase-shifter flexibility \cite{isachbf3}. These studies show that \ac{hbf} can approach the performance of fully digital beamforming under favorable hardware configurations. Nevertheless, its achievable \acp{dof} are still fundamentally limited by fixed antenna positions and fixed radiation characteristics, which becomes particularly restrictive in \ac{isac} systems requiring simultaneous support for \ac{mu} communication and multi-target sensing.

	To further enlarge the design space, recent \ac{isac} research has moved beyond conventional \ac{hbf} toward antenna-domain reconfigurability. Among the emerging architectures, fluid \cite{fa}, movable \cite{ma}, and pinching \cite{pa} antennas belong to the category of position-reconfigurable antennas, while rotatable antennas \cite{6dma} provide orientation reconfigurability. Building on these hardware concepts, recent works have started to incorporate such reconfigurable antennas into \ac{isac} systems. For fluid antennas, early studies considered FA-assisted \ac{isac} and showed that jointly optimizing antenna positions and transmit beamforming can improve both communication and sensing performance \cite{fluidISAC0}. This line has quickly evolved toward more explicit trade-off-oriented designs, where fluid antennas are shown to enlarge the achievable communication--sensing region compared with conventional fixed-position arrays \cite{fluidISAC1}. For movable antennas, recent \ac{isac} designs have investigated bistatic multi-user systems in which the antenna coefficients and antenna positions are jointly optimized to enhance flexible beamforming performance \cite{maISAC}. For pinching antennas, very recent works have introduced PASS-enabled \ac{isac} frameworks, where pinching beamforming is exploited to establish reliable communication and sensing links, and the corresponding communication--sensing rate region has also been characterized \cite{passISAC1,passISAC2}. However, these position-reconfigurable architectures mainly enhance performance through geometry adaptation in the spatial domain, rather than electromagnetic control at the element level. 
	

	Another emerging direction is \ac{em}-domain reconfigurability for \ac{isac}. In our earlier conference work \cite{ChenTriHybridConf}, a \ac{ra}-aided \ac{thbf} framework was introduced, where digital, analog, and \ac{em} beamforming are jointly designed to enlarge the achievable communication-and-sensing trade-off region. More recently, compound reconfigurable antenna arrays were investigated for \ac{isac} clutter suppression \cite{CRAclutterISAC}, where radiation pattern and polarization states are jointly exploited in the \ac{em} domain to enhance sensing performance. Reconfigurable pixel antennas have also been incorporated into \ac{isac} systems, where binary-controlled pixel states provide additional \ac{em}-domain flexibility for joint communication and sensing optimization \cite{pixelISAC}. In parallel, polarization-reconfigurable antennas have recently been introduced into polarimetric \ac{isac}, where polarization agility is exploited to improve fairness while reducing \ac{rf}-chain cost \cite{polarISAC}. These works collectively indicate that recent progress in \ac{isac} is increasingly driven by new antenna architectures that introduce additional \ac{em}-domain \ac{dofs}. Nevertheless, the existing literature still has several limitations. Some works focus on highly specialized hardware freedoms, such as pixel-state or polarization reconfiguration, instead of more general pattern-domain beamforming. Some studies are developed under simplified sensing settings, such as single-target scenarios. In addition, certain works emphasize specific functionalities, such as clutter suppression or fairness, rather than pursuing a more general communication-and-sensing trade-off design. These limitations motivate a more systematic \ac{ra}-enabled \ac{isac} framework that can accommodate different levels of \ac{em}-domain reconfigurability under richer sensing environments.
	
	Motivated by the above observations, and following the \ac{ra} modeling frameworks established in communication-oriented studies in \cite{RA1,pingjunJournal,pingjunsurvey}, we consider two complementary \ac{ra} models for \ac{isac} systems, where \ac{thbf} is exploited for 
	both communication and sensing components in terms of \ac{sinr} enhancement, improved directional energy focusing, and interference mitigation between users and sensing targets. Specifically, the main contributions of this work can be summarized as follows
	\begin{itemize}
		\item We present one of the early systematic investigations of \ac{ra}-enabled \ac{isac} with \ac{thbf}. Specifically, we extend two representative \ac{em}-domain \ac{ra} modelings, as in \cite{RA1,pingjunJournal}, to a monostatic \ac{isac} system. The first one adopts a continuous spherical-harmonic-based radiation representation for \emph{arbitrary pattern generation}, which establishes an ideal upper bound on \ac{em}-domain beamforming flexibility. The second one utilizes radiation patterns data from realistic measurements in \cite{pingjunAntData}, yielding a physically realizable \emph{discrete pattern selection} design. Furthermore, we adopt \ac{scnr} and sum-rate as the objective for sensing and communication in subsequent optimization process.

		\item For the \ac{sust} scenario, we formulate a \ac{fp} optimization problem and apply Dinkelbach's transform to convert it into a convex form. A \emph{closed-form solution} is derived for the fully digital beamforming case, and the corresponding \ac{hbf} solutions are obtained via \ac{mo}. For the \ac{em}-domain beamforming of \ac{ra}, we first perform joint optimization across all antenna elements using \ac{sdr}, and rigorously prove the tightness of the SDR solution. To further reduce computational complexity, we propose an antenna-wise \ac{ao} strategy and derive corresponding closed-form updates for the arbitrary pattern generation model. For the discrete pattern selection model, we traverse all candidate radiation patterns per antenna to obtain the optimal pattern configuration.
		
		\item Without altering the fundamental optimization framework, we extend the proposed design to \ac{mumt} scenarios. Since the objective function becomes a sum-of-ratios form and is no longer directly amenable to Dinkelbach's method, we adopt a Lagrangian dual transform combined with a modified quadratic transform to decouple communication and sensing fractional components. These transforms yield solvable first-order optimality conditions for each user’s beamformer, enabling closed-form fully digital solutions and hybrid approximations similar to the \ac{sust} case. The \ac{em} beamforming for \ac{ra} can be handled in a similar manner. For the arbitrary pattern generation model, either \ac{sdr}-based joint optimization or low-complexity alternating closed-form updates can be employed, while for discrete pattern selection model, only alternating traversal over the finite radiation-pattern set is considered.
	\end{itemize}

Our simulation results demonstrate that, compared to traditional \ac{hbf} schemes with non-reconfigurable antennas, the proposed \ac{ra}-based \ac{thbf} approaches, including both arbitrary pattern generation and discrete pattern selection models, achieve significantly better objective function values in both \ac{sust} and \ac{mumt} scenarios. From the communication sum rate and sensing \ac{scnr} trade-off curves, it can be seen that the proposed method strikes a better trade-off relationship between communication and sensing, substantially enhancing the flexibility of \ac{isac} systems. Furthermore, the provided transmit patterns clearly illustrate that, both at the single-antenna radiation pattern level and the array beampattern level, the \ac{ra}-based schemes designed in this paper are more focused on the communication \acp{cu} and sensing targets. Finally, by examining the objective function values versus RF chain and antenna numbers, we show that the proposed RA-based schemes can achieve the performance of conventional fully digital and large-array systems while utilizing fewer RF chains and antennas, thus significantly reducing hardware overhead.
	
	\emph{Notations:}
	In this paper, we use lower-case letters, lower-case bold letters, and capital bold letters to denote scalars, vectors and matrices respectively. The operators for vectorization, transpose, conjugate, conjugate transpose, inverse, and Moore–Penrose inverse are denoted by  ${\rm vec}\left(\cdot\right)$, $\left(\cdot\right)^{\text{T}}$, $\left(\cdot\right)^{\text{H}}$, $\left(\cdot\right)^{\text{*}}$, $\left(\cdot\right)^{-1}$, and $\left(\cdot\right)^{\dag}$  respectively. $\text{Tr}\left(\mathbf{A}\right)$ and $\text{Rank}\left(\mathbf{A}\right)$ stand for the trace and rank of matrix $\mathbf{A}$. The $i$th element of vector $\mathbf{a}$ is $[\mathbf{a}]_i$, the $i$th column of matrix $\mathbf{A}$ is $\left[\mathbf{A}\right]_{:,i}$, and the $(i,j)$th element of matrix $\mathbf{A}$ is $[\mathbf{A}]_{i,j}$. $|a|$, $\|\mathbf{a} \|_2$, $\|\mathbf{A} \|_2$, and $\|\mathbf{A} \|_F$ respectively represent the modulus of scalar $a$, $\ell$-2 norm of vector $\mathbf{a}$, induced 2-norm of matrix $\mathbf{A}$, and Frobenius norm of matrix $\mathbf{A}$. Given a scalar $a$, $a!$ denotes the factorial operation on $a$. {\color{black}$\mathbb{H}^+$ denotes the Hermitian semipositive definite matrix set. The union of sets $\mathcal{A}$ and $\mathcal{B}$ is represented as $\mathbf{A}\cup \mathcal{B}$. The set $\mathcal{A}$ with the element $a$ removed is represented as $\mathcal{A} \setminus a$.} The operations for extracting the imaginary part and real part of a complex variable are denoted by $\Im\left(\cdot\right)$ and $\Re\left(\cdot\right)$ respectively. $\otimes$ represents the Kronecker product. Moreover, $\jmath = \sqrt{-1}$ is the imaginary unit. $\mathbf{I}_N$ denotes the $N$-dimensional unit array. Considering an optimization problem with respect to $x$, the optimal solution is denoted by $x^{\star}$.

	\section{System Model}

	\begin{figure}[t]
		\includegraphics[scale=0.7]{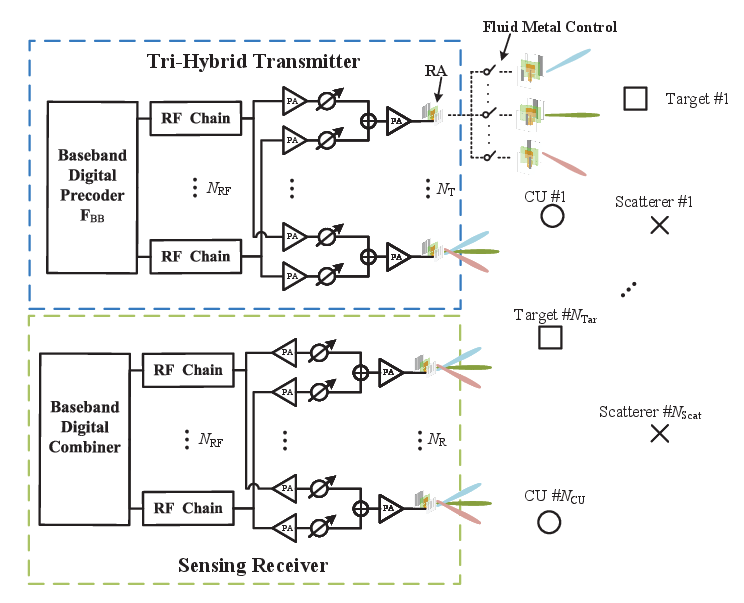}
		\vspace{-1mm}
		\caption{System model of the proposed \ac{isac} system with \ac{ra}.}\label{systemmodel}
		\vspace{-3mm}
	\end{figure}
	
	As illustrated in Fig. \ref{systemmodel}, we consider a downlink  \ac{miso} \ac{isac} system consisting of a transmitter equipped with an $N_{\rm T}$-element \ac{upa}, an $N_{\rm R}$-element UPA radar receiver. For the communication part, we consider $N_{\rm CU}$ single-antenna \acp{cu} located at elevation and azimuth angle pairs $(\theta^{\rm cu}_{i}, \phi^{\rm cu}_{i})$. For the sensing part, we consider $N_{\rm Tar}$ and $N_{\rm Scat}$ point-like sensing targets and scatterers, positioned at $(\theta^{\rm Tar}_i, \phi^{\rm Tar}_i)$ and $(\theta^{\rm Scat}_i, \phi^{\rm Scat}_i)$. The transmitter adopts a \ac{thbf} architecture with $N_{\rm RF}$  \ac{rf} chains and $N_{\rm T}$ \acp{ra}. For the sensing functionality, the receiver is assumed to have a symmetric architecture of the same type. Exploiting this transceiver symmetry, we restrict the optimization to the transmit-side beamforming and keep the receive-side beamforming fixed. When desired, the receive combiner can be optimized in a similar alternating fashion following the same methodology as for the transmitter.

	\subsection{Transmission Model}
	First, the transmitting signal vector $\mathbf{x} \in \mathbb{C}^{N_{\rm T}}$ over the $N_{\rm T}$ antennas is formulated as
	\begin{equation}
	\mathbf{x} = \mathbf{F}_{\rm RF} \mathbf{F}_{\rm BB} \mathbf{s},
	\end{equation} 
	where 	$\mathbf{s} \in \mathbb{C}^{N_{\rm CU}}$ denotes the  transmitting data symbol vector for communication, satisfying $\mathbb{E}[\mathbf{s}\mathbf{s}^{\rm H}] = \mathbf{I}_{N_{\rm CU}}$. Matrices $\mathbf{F}_{\rm RF} \in \mathbb{C}^{N_{\rm T} \times N_{\rm RF}}$ and $\mathbf{F}_{\rm BB} \in \mathbb{C}^{N_{\rm RF} \times N_{\rm CU}}$ represent the analog and digital beamformer, respectively. The analog beamformer is implemented by a fully-connected phase-shifter network with unit-modulus constraints as $[\mathbf{F}_{\rm RF}] = 1, \forall i,\;j$. The analog and digital beamforming needs to obey a total transmitting power budget $P_{\rm T}$ as
	\begin{equation}
	 \left\| \mathbf{F}_{\rm RF} \mathbf{F}_{\rm BB} \right\|_{\rm F}^2 \leq P_{\rm T}. \notag
	\end{equation}

	Let  $\mathbf{H}_{{\rm t}_i}$, $\mathbf
	{H}_{{\rm s}_j}$, and $\mathbf{h}_{{\rm c}_k}$ denote the channels relating to the $i$th target, the $j$th scatterer, and $k$th \ac{cu}. The signals received by the $k$th \ac{cu} and the sensing receiver can be expressed as
	\begin{equation}\label{reccom}
	{y}_{{\rm c}_k} = \mathbf{h}_{{\rm c}_k}^{\rm H}  \mathbf{F}_{\rm RF} [\mathbf{F}_{\rm BB} ]_{:,k} \mathbf{s} + \sum_{j \neq k} \mathbf{h}_{{\rm c}_k}^{\rm H}  \mathbf{F}_{\rm RF} [\mathbf{F}_{\rm BB} ]_{:,j}\mathbf{s} + n_k,
	\end{equation}
	\begin{equation}\label{recsen}
	\mathbf{y}_{\rm r} = ( \sum_{i = 1}^{N_{\rm Tar}} \mathbf{H}_{{\rm t}_i}^{\rm H} + \sum_{j = 1}^{N_{\rm Scat}} \mathbf{H}_{{\rm s}_j}^{\rm H} )\mathbf{x} + \mathbf{n}_{\rm r},
	\end{equation}
	where $n_k$ is the \ac{awgn} at the $k$th \ac{cu}, obeying $\mathcal{CN}(0,\sigma_{\rm n}^2)$, and $\mathbf{n}_{\rm r}$ denotes the \ac{awgn} vector at the radar receiver, obeying $\mathcal{CN}(0,\sigma_{\rm n}^2\mathbf{I}_{N_{\rm R}})$.

	In the conventional \ac{hbf} framework, we optimize the analog beamformer $\mathbf{F}_{\rm RF}$ and digital beamformer $\mathbf{F}_{\rm BB}$ for the improvement of transmission rate. In contrast, the reconfigurability of transmit antennas brings additional \ac{dofs} to change channels $\mathbf{H}_{{\rm t}_i}$, $\mathbf
	{H}_{{\rm s}_j}$, $\mathbf{h}_{{\rm c}_k}$. In the following subsection, two \ac{ra} channel models are introduced.
	The first one is an arbitrary pattern generation model based on spherical harmonic theory, while the second one is a discrete pattern selection model derived from measured data in \cite{pingjunAntData}.
	
	Since the sensing receiver is also assumed to employ a \ac{thbf} architecture, its overall combining vector for the $i$th target is written as
	\begin{equation}
		\mathbf{w}_i^{\rm H}
		=
		\mathbf{w}_{{\rm BB},i}^{\rm H}
		\mathbf{W}_{{\rm RF}}^{\rm H}
		\mathbf{W}_{{\rm EM}}^{\rm H},
	\end{equation}
where $\mathbf{W}_{{\rm EM}}$, $\mathbf{W}_{{\rm RF}}$, and $\mathbf{w}_{{\rm BB},i}$ denote the \ac{em}, analog, and digital combiners for the $i$th target, respectively. Accordingly, the post-combining sensing signal associated with the $i$th target is given by
\begin{equation}\label{recsen_combined}
	\tilde y_{{\rm r},i}
	=
	 \mathbf{w}_i^{\rm H} ( \sum_{i = 1}^{N_{\rm Tar}} \mathbf{H}_{{\rm t}_i}^{\rm H} + \sum_{j = 1}^{N_{\rm Scat}} \mathbf{H}_{{\rm s}_j}^{\rm H} )\mathbf{x} + 	 \mathbf{w}_i^{\rm H} \mathbf{n}_{\rm r}.
\end{equation}

	\subsection{\ac{ra} Channel Model} \label{rachanmodel}
	According to \cite{pingjunJournal}, under far-field conditions, a general channel model applicable to arbitrary antenna types and array shapes can be expressed as
	\begin{equation}
		\mathbf{h} = \sqrt{ \frac{N_{\rm T}}{L}}\sum_{l = 1}^{L} {\alpha}_{l} \cdot \mathbf{g} \odot \mathbf{a}_l^{N_{\rm T}}, 
	\end{equation}
	where $L$ is the path number, $\alpha_{l}$ denotes the corresponding complex channel gain \cite{RCS}. The \ac{agv} $\mathbf{g}\in \mathbb{C}^{N_{\rm T}}$ and a \ac{upa}-based \ac{arv} $\mathbf{a}^{N_{\rm T}}_{l}\in \mathbb{C}^{N_{\rm T}}$ are defined as
	    \begin{equation}
		\mathbf{g} = [G^{(1)}\left(\theta,\phi\right),\,\dots,G^{(N_{\rm T})}\left(\theta,\phi\right)]^{\rm T},
	\end{equation}
	    \begin{equation}
		\begin{aligned}
			\mathbf{a}^{N_{\rm T}}_{l}\left(\theta_{l},\phi_{l}\right) = e^{-\jmath 2 \pi / \lambda [0:N_{\rm T}^x-1] d_x \sin \theta_{l} \cos \phi_{l} } \otimes \\
			e^{-\jmath 2 \pi/\lambda [0:N_{\rm T}^y-1] d_y \sin \theta_{l} \sin \phi_{l} },
		\end{aligned}
	\end{equation}
where $N_{\rm T} = N_{\rm T}^x N_{\rm T}^y$, $N_{\rm T}^x$ and $N_{\rm T}^y$ are the antenna number along $x$ and $y$ axis, $d_x$ and $d_y$ are the corresponding antenna element spacing.
	\subsubsection{\ac{ra} Model I: Arbitrary Pattern Generation}\label{ChanModelI}
	In this part, an \ac{ra} model is introduced using spherical harmonic theory. Focusing on the specific expression of the \ac{agv}, according to \cite{RA1}, the radiation pattern of the $n$th \ac{ra} can be decomposed into an infinite summation of spherical harmonics as follows
	\begin{equation}\label{NoAppG}
		G^{(n)} = G^{(n)}\left(\theta,\phi\right) = \sum_{u=0}^{+\infty} \sum_{q = -u}^{u} c_{uq}^{(n)}Y_u^{q}\left(\theta,\phi\right),
	\end{equation}
	where $u$ represents the degree of the spherical harmonic function, $q$ denotes the order of the spherical harmonic function. $c_{uq}^{(n)}$ is the corresponding harmonic coefficient, $Y_u^{q}$ denotes the spherical harmonics, given by
	\begin{equation}
		Y_u^q(\theta,\phi) = N_u^{q}
		P_u^{q}(\cos\theta)\; e^{\jmath {q} \phi},
		\; u \ge 0,\; -u \le q \le u ,
	\end{equation}
	where $N_u^{q} = \sqrt{\frac{2u + 1}{4 \pi} \frac{(u-|q|)!}{(u+|q|)!}}$ is the normalization factor, $P_u^{q}(\cos\theta)$ represents the associated Legendre function of the $u$th degree and $q$th order. To make the analysis tractable, we follow the truncation operation in \cite{RA1} with a window length $T = U^2 + 2U + 1$. Therefore, (\ref{NoAppG}) can be approximated as
	\begin{equation}\label{AppG}
		G^{(n)}(\theta,\phi)\approx\sum_{u=0}^{U}\sum_{q=-u}^{u}c_{uq}^{(n)}Y_{u}^{q}(\theta,\phi)=\sum_{t=1}^{T} \dot{c}_{t}^{(n)} \dot {Y}_{t}(\theta,\phi),
	\end{equation}
	where $\dot{c}^{(n)}_t = c_{uq}^{(n)}$ and $\dot{Y}_t\left(\theta,\phi\right) = Y_u^q \left(\theta,\phi\right)$ with $t = u^2 + u + q + 1, \, u \in [0,U], \, q \in [-u,u]$. In the rest of this paper, we only consider the truncation signal model under the notation of $\mathbf{\bar c}^{(n)} \triangleq [\dot{c}_{1}^{(n)}, \dot{c}_{2}^{(n)}, \ldots, \dot{c}_{T}^{(n)}]^{\rm T} \in \mathbb{C}^{T}$, and $\mathbf{b}(\theta, \phi) \triangleq [\dot{Y}_{1}(\theta, \phi),  \ldots, \dot{Y}_{T}(\theta, \phi)]^{\rm T} \in \mathbb{C}^{T}$. Thus, (\ref{AppG}) can be rewritten in a matrix form as
	\begin{equation}\label{SynRAI}
	G^{(n)}(\theta,\phi)\approx \mathbf{b}^{\rm T}(\theta, \phi) \mathbf{\bar c}^{(n)}.
	\end{equation}
	According to the energy conservation law in \cite{RA1}, an energy constraint on $\mathbf{\bar c}^{(n)}$ is introduced as 
	\begin{equation}
		\|\mathbf{\bar c}^{(n)}\|_2^2 = 1 , \forall n.
	\end{equation}
	\subsubsection{\ac{ra} Model II: Discrete Pattern Selection}
	The \ac{ra} model based on spherical harmonics in \ref{ChanModelI} is an idealized model. The synthesized radiation patterns are not always physically realizable via practical antenna hardware, thus serving as a performance upper bound for \ac{ra}s. To obtain a more realistic characterization, this part introduces an \ac{ra} model based on measured data reported in \cite{pingjunAntData}. Specifically, there exist $S$ pre-fabricated radiation patterns, represented as $\mathcal{G} = \left\{ {{{\tilde G}_1},{{\tilde G}_2}, \cdots ,{{\tilde G}_S}} \right\}$. Each element in $\mathcal{G}$ is actually a function with respect to $\theta$ and $\phi$. For given angle pair $(\theta,\phi)$, we define 
	\begin{equation}
		{\mathbf{\tilde b}}\left( {\theta ,\phi } \right) = {\left[ {{{\tilde G}_1}\left( {\theta ,\phi } \right),{{\tilde G}_2}\left( {\theta ,\phi } \right), \cdots ,{{\tilde G}_S}\left( {\theta ,\phi } \right)} \right]^{\text{T}}}.
	\end{equation} 
	Similar to (\ref{SynRAI}), a radiation pattern can be synthesized as
	\begin{equation}
	{G^{(n)}}(\theta ,\phi ) = {{{\mathbf{\tilde b}}}^{\text{T}}}(\theta ,\phi ){{{\mathbf{\tilde c}}}^{(n)}},
	\end{equation}
	where ${{{\mathbf{\tilde c}}}^{(n)}}$ is a binary selection vector for the $n$th \ac{ra}. Mathematically, this constraint can be formulated as
	\begin{equation}\label{selecon}
		{{{\mathbf{\tilde c}}}^{(n)}} \in \left\{ {{\mathbf{c}} \in {{\left\{ {0,1} \right\}}^S}|{{\left\| {\mathbf{c}} \right\|}_0} = 1} \right\},\;\forall n.
	\end{equation} 
	Additionally, the power constraint on the radiation patterns of the \ac{ra} is inherently satisfied, since this constraint is explicitly captured in the process of generating the radiation-pattern dataset $\mathcal{G}$.

\subsubsection{Communication Channel Model}
Based on Section \ref{rachanmodel}, the communication channel vector toward the $k$th \ac{cu} can be further formulated in a compact format as
\begin{equation}
	  \mathbf{h}_{{\rm c}_k}  = \mathbf{F}_{\rm EM}^{\rm H} \mathbf{h}_{{\rm c}_k}^{\rm EM} = \sqrt{\frac{N_{\rm T}}{L_k}} \sum_{l=1}^{L_k} \mathbf{F}_{\rm EM}^{\rm H} \mathbf{h}_{{\rm c}_{k,l}}^{\rm EM},
\end{equation}
where $L_k$ denotes the number of paths associated with $k$th \ac{cu}, and the $l$th \ac{aod} pair for the $k$th \ac{cu} is represented as $(\theta_{{\rm c}_{k,l}},\phi_{{\rm c}_{k,l}})$. Furthermore, taking \ac{ra} Model I as an example, $\mathbf{F}_{\rm EM}$ is the \ac{em} beamforming matrix, expressed as
\begin{equation}
	\mathbf{F}_{\rm EM} = {\rm blkdiag} \left\{ \mathbf{\bar c}^{(1)}, \, \mathbf{\bar c}^{(2)}, \, \cdots, \, \mathbf{\bar c}^{(N_{\rm T})} \right\}.
\end{equation}
Furthermore, $ \mathbf{h}_{{\rm c}_{k,l}}^{\rm EM}$ is given by
\begin{equation}
	\mathbf{h}_{{\rm c}_{k,l}}^{\rm EM} =  \mathbf{b}\left(\theta_{{\rm c}_{k,l}},\phi_{{\rm c}_{k,l}}\right) \otimes \mathbf{1}_{N_{\rm T}}  \odot \left(	\alpha_{{\rm c}_{k,l}}  \cdot \mathbf{a}^{N_{\rm T}}_{{\rm c}_{k,l}} \otimes \mathbf{1}_{T}\right),
\end{equation}
where $\alpha_{k,l}$ denotes the corresponding complex channel gain \cite{RCS}, defined as
\begin{equation}
	\alpha_{{\rm c}_{k,l}} = \sqrt{\frac{\lambda^2}{(4\pi)^2 (r_{{\rm c}_{k,l}})^{\sigma_{{\rm PL}}}} }e^{\jmath \psi_{{\rm c}_{k,l}}}.
\end{equation}
Here, $\lambda$ denotes the wavelength, $\sigma_{\rm PL}$ is the path loss exponent, $r_{{\rm c}_{k,l}} $ is the propagation distance, and $\psi_{{\rm c}_{k,l}}$ is a random phase shift.

For Model II, the above expressions remain the same, except that the coefficients $T$, $\mathbf{c}^{(n)}$, and $\mathbf{b}$ are replaced by $S$, $\mathbf{\tilde c}^{(n)}$, and $\mathbf{\tilde b}$, respectively.

\subsubsection{Sensing Channel Model}
For the round-trip sensing channel model, taking the $i$th target as an example, the effective sensing channel after transmit and receive \ac{em}-domain processing can be expressed as
\begin{equation}
	\begin{aligned}
	\mathbf{H}_{{\rm t}_i}
	=
	\mathbf{F}_{\rm EM}^{\rm H}
	\mathbf{H}_{{\rm t}_i}^{\rm EM}
	\mathbf{W}_{{\rm EM}}
	=
	\mathbf{F}_{\rm EM}^{\rm H}
	\mathbf{h}_{{\rm t}_i}^{\rm EM}
	\left(	\widetilde{\mathbf{h}}_{{\rm t}_i}^{{\rm EM }}\right)^{\rm H}	\mathbf{W}_{{\rm EM}}.
	\end{aligned}
\end{equation}
The backward channel can be expressed as $\widetilde{\mathbf{h}}_{{\rm t}_i} = \mathbf{W}_{{\rm EM}}^{\rm H} \widetilde{\mathbf{h}}_{{\rm t}_i}^{{\rm EM }}	$, where
\begin{equation}
\widetilde{\mathbf{h}}_{{\rm t}_i}^{{\rm EM }}= \mathbf{b}\left(\theta_{{\rm t}_i},\phi_{{\rm t}_i}\right) \otimes \mathbf{1}_{N_{\rm R}}  \odot \left(\sqrt{{\alpha}_{{\rm t}_i}} \cdot \mathbf{a}_{{\rm t}_i}^{N_{\rm R}} \otimes \mathbf{1}_{T}\right).
\end{equation}
The forward channel is given by $\mathbf{h}_{{\rm t}_i}= \mathbf{F}_{\rm EM}^{\rm H} \mathbf{h}_{{\rm t}_i}^{\rm EM}$, where 
\begin{equation}
	\mathbf{h}_{{\rm t}_i}^{\rm EM} = \mathbf{b}\left(\theta_{{\rm t}_i},\phi_{{\rm t}_i}\right) \otimes \mathbf{1}_{N_{\rm T}}  \odot \left(\sqrt{{\alpha}_{{\rm t}_i} } \cdot \mathbf{a}_{{\rm t}_i}^{N_{\rm T}} \otimes \mathbf{1}_{T}\right).
\end{equation}
$\alpha_{{\rm t}_i}$ denotes the sensing channel gain modeled by the round trip path loss and radar cross-section \cite{RCS} as
\begin{equation}
{\alpha _{{{\text{t}}_i}}} = \sqrt{\frac{{{\lambda ^2}\sigma _{{{\text{t}}_i}}^{{\text{RCS}}}}}{{{{\left( {4\pi } \right)}^3}r_{{{\text{t}}_i}}^4}}}{e^{\jmath{\psi _{{{\text{t}}_i}}}}},
\end{equation}
here ${\sigma _{{{\text{t}}_i}}^{{\text{RCS}}}}$ denotes the radar cross-section, $r_{{\rm t}_i}$ represents the distance of the $i$th target, and $\psi_{{\rm t}_{i}}$ is a random phase shift.

\subsection{Performance Metrics}
Based on the receive communication signal in (\ref{reccom}) and the \ac{ra} channel model, the \ac{sinr} $\gamma_k$ of the $k$th \ac{cu} is formulated as
\begin{equation}\label{comsinr}
{{\gamma} _k} = \frac{{{A_k}}}{{{B_k}}} = \frac{{{{\left| {{\mathbf{h}}_{c_k}^{{\text{EM}}\;{\text{H}}}{{\mathbf{F}}_{{\text{EM}}}}{{\mathbf{F}}_{{\text{RF}}}}{{\left[ {{{\mathbf{F}}_{{\text{BB}}}}} \right]}_{:,k}}} \right|}^2}}}{{\sum\limits_{j \ne k} {{{\left| {{\mathbf{h}}_{c_k}^{{\text{EM}}\;{\text{H}}}{{\mathbf{F}}_{{\text{EM}}}}{{\mathbf{F}}_{{\text{RF}}}}{{\left[ {{{\mathbf{F}}_{{\text{BB}}}}} \right]}_{:,j}}} \right|}^2} + \sigma _{\text{n}}^2} }}.
\end{equation}
Accordingly, the communication sum rate is given by
\begin{equation}
	{R_{{\rm c}}} = \sum\limits_{k = 1}^K R_k =\sum\limits_{k = 1}^K {{{\log }_2}\left( {1 + {{\gamma}_{{{{k}}}}}} \right)}.
\end{equation}

Since joint transmitter-receiver optimization is usually difficult under the \ac{thbf} architecture, we adopt a sequential design strategy that first optimizes the transmit-side beamforming with the receive beamformer fixed. Such a transmitter-first design approach has been widely used in the \ac{hbf} field as an effective way to reduce the complexity of joint transceiver optimization \cite{HBFYu}. Thus, the receive \ac{scnr} for the $i$th target can be expressed as
\begin{equation}\label{sensinr}
\begin{aligned}
	{\eta _i} = \frac{{{C_i}}}{{{D_i}}} &= \frac{{\left\| {{\mathbf{w}}_i^{\text{H}}{\mathbf{\tilde h}}_{{{\text{t}}_i}}^{{\text{EM}}}{\mathbf{h}}_{{{\text{t}}_i}}^{{\text{EM}}\;{\text{H}}}{{\mathbf{F}}_{{\text{EM}}}}{{\mathbf{F}}_{{\text{RF}}}}{{\mathbf{F}}_{{\text{BB}}}}} \right\|_2^2}}{{\sum\limits_{\kappa  \in {\mathcal{I}_i}} {\left\| {{\mathbf{w}}_i^{\text{H}}{\mathbf{\tilde h}}_\kappa ^{{\text{EM}}}{\mathbf{h}}_\kappa ^{{\text{EM H}}}{{\mathbf{F}}_{{\text{EM}}}}{{\mathbf{F}}_{{\text{RF}}}}{{\mathbf{F}}_{{\text{BB}}}}} \right\|_2^2}  + \tilde \sigma _{\text{n}}^2}} \hfill \\
	&= \frac{{{\omega _i}\left\| {{\mathbf{h}}_{{{\text{t}}_i}}^{{\text{EM}}\;{\text{H}}}{{\mathbf{F}}_{{\text{EM}}}}{{\mathbf{F}}_{{\text{RF}}}}{{\mathbf{F}}_{{\text{BB}}}}} \right\|_2^2}}{{\sum\limits_{\kappa  \in {\mathcal{I}_i}} {{\omega _\kappa }\left\| {{\mathbf{h}}_\kappa ^{{\text{EM H}}}{{\mathbf{F}}_{{\text{EM}}}}{{\mathbf{F}}_{{\text{RF}}}}{{\mathbf{F}}_{{\text{BB}}}}} \right\|_2^2}  + \tilde \sigma _{\text{n}}^2}} \hfill \\ 
\end{aligned} 
\end{equation}
where $\omega_i = | {{\mathbf{w}}_i^{\text{H}}{\mathbf{\tilde h}}_{{{\text{t}}_i}}^{{\text{EM}}}}|^2$, $\omega_{\kappa} = | {{\mathbf{w}}_i^{\text{H}}{\mathbf{\tilde h}}_\kappa ^{{\text{EM}}}}|^2$, $\tilde{\sigma}_{\rm n}^2 = \sigma_{\rm n}^2 \|\mathbf{w}_i\|^2$. $\mathcal{I}_i$ denotes the interference channel set of the $i$th target, including other targets and all scatterers as
\begin{equation}
	\mathcal{I}_i = \left\{ {\left\{ {{{\text{t}}_{i^{'}}}} \right\}_{i^{'} = 1}^{{N_{{\text{Tar}}}}}\backslash \left\{ {{{\text{t}}_i}} \right\}} \right\} \cup \left\{ {{{\text{s}}_j}} \right\}_{j = 1}^{{N_{{\text{Scat}}}}}.
\end{equation}

\section{Tri-hybrid ISAC Beamforming Design under \ac{sust} Scenarios}
	In this section, we present the \ac{thbf} design for \ac{ra}-\ac{isac} systems under \ac{sust} scenarios, based on both the \ac{ra} models I and II. In this section, since $N_{\rm CU} = 1$, the matrix $\mathbf{F}_{\rm BB}$ is rewritten in vector form as $\mathbf{f}_{\rm BB}$. Furthermore, the indices $k$ and $i$ in (\ref{comsinr}) and (\ref{sensinr}), originally used to distinguish \ac{cu}s and targets, are removed for ease of notation in this section.

	\subsection{ISAC Beamforming Problem Formulation}
	Under the constraints of transmit power budget, antenna pattern power, and unit-modulus of the phase shifters, the baseband digital beamformer $\mathbf{f}_{\rm BB} \in \mathbb{C}^{N_{\rm RF} \times 1}$, analog beamformer $\mathbf{F}_{\rm RF} \in \mathbb{C}^{N_{\rm T} \times N_{\rm RF}}$, and \ac{em} beamformer $\mathbf{F}_{\rm EM} \in \mathbb{C}^{N_{\rm T} (S \; {\rm or}\;T) \times N_{\rm T}}$ are jointly optimized to maximize the weighted sum of the communication \ac{snr} \footnote{For single user case, rate maximization and \ac{snr} maximization is similar in meaning, differing only in scale.} and sensing \ac{scnr} as
	\begin{subequations}\label{optsust}
		\begin{align}
			&\mathcal{P}_{1}:\;\mathop {\max }\limits_{{\mathbf{F}_{\rm EM}},{\mathbf{F}_{\rm RF}},{\mathbf{f}_{\rm BB}}} \tilde{\beta} {\gamma} + \beta \eta\\
			&{\rm s.t.} \ \ {\rm Tr}\left({{\mathbf{F}}_{{\text{RF}}}}{{\mathbf{f}}_{{\text{BB}}}}{\mathbf{f}}_{{\text{BB}}}^{\text{H}}{\mathbf{F}}_{{\text{RF}}}^{\text{H}}\right) \leq P, \\
			&\ \ \ \ \  \ {\rm RA \;constraints \; w.r.t. } \mathbf{F}_{\rm EM},
		\end{align}%
	\end{subequations}
	where $\tilde{\beta}$ and $\beta$ are trade-off factors between communication and sensing, satisfying $\tilde{\beta} + \beta = 1$. 
	To make this problem tractable, \ac{ao} is employed between the \ac{hbf} and the \ac{em} beamforming. 
	
	\subsection{Optimizing $\mathbf{F}_{\rm RF}$ and $\mathbf{f}_{\rm BB}$ with given $\mathbf{F}_{\rm EM}$}
	For a given \ac{em} beamformer, Problem (\ref{optsust}) becomes
	\begin{subequations}\label{optsustbf}
		\begin{align}
			&\mathcal{P}_{1.1}:\;\mathop {\max }\limits_{{\mathbf{F}_{\rm RF}},{\mathbf{f}_{\rm BB}}} \tilde{\beta} {\gamma} + \beta \eta\\
			&{\rm s.t.} \ \left\| {{{\mathbf{F}}_{{\text{RF}}}}{{\mathbf{f}}_{{\text{BB}}}}} \right\|_2^2 \leqslant P,\\
			&\ \ \ \ \     |[\mathbf{F}_{\rm RF}]_{i,j}| = 1, \forall i,j.
		\end{align}
	\end{subequations}
	To make this problem tractable, we first consider its corresponding digital beamforming problem by putting $\mathbf{F}_{\rm RF}$ and $\mathbf{f}_{\rm BB}$ together as $\mathbf{f}_{\rm FD} \triangleq \mathbf{F}_{\rm RF} \mathbf{f}_{\rm BB} $ and dropping the unit-modulus constraint temporarily. Then, Dinkelbach's transform is  employed here to reformulate the single-ratio \ac{fp}. First, we note that the fractional term $\eta$ in the objective function is not quasi-concave, as its numerator is concave while the denominator is convex. This structural property generally makes the resulting fractional programming problem analytically intractable and challenging to solve directly. To address this, we present the following proposition, which shows that the global optimal solution of $\mathbf{f}_{\rm FD}$ can still be efficiently obtained via Dinkelbach's transform despite the non-quasi-concavity of $\eta$.
	\begin{proposition}\label{Pro0}
	The global optimal solution for $\mathbf{f}_{\rm FD}$ can be achieved via Dinkelbach's transform, even though $\eta$ is not quasi-concave
	\end{proposition}
		\begin{proof}
		See Appendix \ref{AppdixA} for the proof.
	\end{proof}
	
	Therefore, the optimization problem \ac{wrt} $\mathbf{f}_{\rm FD}$ becomes
	
	\begin{subequations}\label{optsustfdbf}
		\begin{align}
			\mathcal{P}_{1.2}:\;&\mathop {\max }\limits_{{\mathbf{f}_{\rm FD}}} {\text{ }}{\mathbf{f}}_{{\text{FD}}}^{\text{H}}{\mathbf{R}}{{\mathbf{f}}_{{\text{FD}}}} -  \varsigma \tilde\sigma _{\text{n}}^2 \\
			&{\rm s.t.} \ \ \left\| {{{\mathbf{f}}_{{\text{FD}}}}} \right\|_2^2 \leqslant P,
		\end{align}%
	\end{subequations}
	where $\mathbf{R}$ is the channel covariance matrix, given by
	\begin{equation}\label{RallmodelI}
		{\mathbf{R}} = \tilde \beta {{\mathbf{R}}_{{\text{EM,c}}}}/\sigma _{\text{n}}^2 +  \beta \omega {{{\mathbf{R}}_{{\text{EM,t}}}} - \varsigma \sum \limits_{\kappa \in \mathcal{I}} \omega_{\kappa} {{\mathbf{R}}_{{\text{EM,}\kappa}}}},
	\end{equation}
	where ${{\mathbf{R}}_{{\text{EM,c}}}}$, ${{\mathbf{R}}_{{\text{EM,t}}}}$, and ${{\mathbf{R}}_{{\text{EM,}\kappa}}}$ are defined in the same form as
	\begin{equation}\label{RmodelI}
		{{\mathbf{R}}_{{\text{EM},\kappa}}} = {\left( {{\mathbf{h}}_{\kappa}^{{\text{EM}}\;{\text{H}}}{{\mathbf{F}}_{{\text{EM}}}}} \right)^{\text{H}}}{\mathbf{h}}_{\kappa}^{{\text{EM}}\;{\text{H}}}{{\mathbf{F}}_{{\text{EM}}}}.
	\end{equation}
	$\varsigma$ is an auxiliary variable introduced by Dinkelbach’s transform \cite{Dinkelbach} for $\eta$, with its optimal solution iterating as 
	\begin{equation}\label{lambda1}
		\varsigma ^ \star  = \beta \omega {\mathbf{f}}_{{\text{FD}}}^{{\text{H}}}{{\mathbf{R}}_{{\text{EM,t}}}}{\mathbf{f}}_{{\text{FD}}}/\left( {{\mathbf{f}}_{{\text{FD}}}^{{\text{H}}}\sum\limits_{\kappa \in \mathcal{I}}  \omega_{\kappa}{{{\mathbf{R}}_{{\text{EM,}}{{\kappa }}}}} {\mathbf{f}}_{{\text{FD}}}+ \tilde \sigma _{\text{n}}^2} \right),
	\end{equation}
	where $\mathcal{I}$ is the interference channel set for the target.

	Since the power budget of $\mathbf{f}_{\rm FD}$ should always be fully exploited to enhance the objective function, the optimal solution of $\mathcal{P}_{1.2}$ can be iterated through the Rayleigh–Ritz method as 
	\begin{equation}\label{fdsolution}
		\mathbf{f}_{\rm FD}^{\star} = \sqrt{P} \mathbf{u},
	\end{equation}
	where $\mathbf{u}$ is the eigenvector corresponding to the maximum eigenvalue of $\mathbf{R}$. Finally, the optimal fully digital beamformer $\mathbf f_{\rm FD} $ is approximated by a hybrid structure $ \mathbf F_{\rm RF}\mathbf f_{\rm BB} $ through the following problem:
	\begin{subequations}
		\begin{align}
			&\mathop {\min }\limits_{\mathbf{F}_{\rm RF},\;\mathbf{f}_{\rm BB}} \|\mathbf{F}_{\rm RF} \mathbf{f}_{\rm BB} - \mathbf{f}_{\rm FD}\|_F^2,\\
			&{\rm{s}}{\rm{.t}}{\rm{.}}\; \left|\left[\mathbf{F}_{\rm RF}\right]_{i,j}\right| = 1,\; \forall i,\;j.
		\end{align}
		\label{MO}%
	\end{subequations}
	The details of solving this hybrid approximation problem can be found in \cite{HBFApp1,HBFApp2} via \ac{mo}, which is omitted here.
	\subsection{Optimizing $\mathbf{F}_{\rm EM}$ with given $\mathbf{F}_{\rm RF}$ and $\mathbf{f}_{\rm BB}$}\label{sustEM}
	\subsubsection{RA Model I}
	The next step is to optimize the EM beamformer with given $\mathbf{f}_{\rm FD}$ in (\ref{fdsolution}). Firstly, the objective can be rewritten with respect to $\mathbf{\bar c} \triangleq [{{{\mathbf{\bar c}}^{(1)\;{\text{T}}}}}, \; \cdots, {{{\mathbf{\bar c}}^{(N_{\rm T})\;{\text{T}}}}}]^{\rm T}$ as
	\begin{equation}\label{objfunc_c}
		{f_{{\text{obj}}}} = \frac{{\tilde \beta }}{{\sigma _{\text{n}}^2}}{{{\mathbf{\bar c}}}^{\text{H}}}{\mathbf{A}}_{{\text{EM,c}}}^{}{\mathbf{\bar c}} + \frac{{ \beta \omega {{{\mathbf{\bar c}}}^{\text{H}}}{\mathbf{A}}_{{\text{EM,t}}}^{}{\mathbf{\bar c}}}}{{\sum\limits_{\kappa \in \mathcal{I}} \omega_{\kappa} {{{{\mathbf{\bar c}}}^{\text{H}}}{{\mathbf{A}}_{{\text{EM,}}{\kappa}}}{\mathbf{\bar c}} +\tilde\sigma _{\text{n}}^2} }}.
	\end{equation}
	Using $\kappa$ to represent the channel type, $\mathbf{A}_{{\rm EM},\kappa} = \mathbf{a}_{{\rm EM},\kappa} \mathbf{a}^{\rm H}_{{\rm EM},\kappa}$, and  $\mathbf{a}_{{\rm EM},\kappa}$ can be expressed as
	\begin{equation}\label{eq:arrangement}
		{\mathbf{a}}_{{\text{EM},{\kappa}}}^{} = {\left[ {\begin{array}{*{20}{c}}
					{{\mathbf{a}}_{{\text{EM},{\kappa}}}^{(1)\;{\text{T}}}}&{{\mathbf{a}}_{{\text{EM},{\kappa}}}^{(2)\;{\text{T}}}}& \cdots &{{\mathbf{a}}_{{\text{EM},{\kappa}}}^{({N_{\text{T}}})\;{\text{T}}}} 
			\end{array}} \right]^{\text{T}}},
	\end{equation}
	where
	\begin{equation}
		{\mathbf{a}}_{{\text{EM,}\kappa}}^{(n)} = {\left[ {{\mathbf{f}}^{*}_{{\text{FD}}}} \right]_n}{\left[ {{\mathbf{h}}_{\kappa}^{{\text{EM}}}} \right]_{\left( {n - 1} \right)T + 1:nT}}.  
	\end{equation}

	Following Proposition \ref{Pro0}, Dinkelbach’s transform can be directly applied here without loss of optimality (see Appendix \ref{AppdixB} for the proof). Therefore, the optimization problem for \ac{em} beamforming can be formulated as
	\begin{subequations}\label{optsustem}
	\begin{align}
		&\mathcal{P}_{1.3}:\;\mathop {\max }\limits_{{\mathbf{\bar c }}} {{{\mathbf{\bar c}}}^{\text{H}}}{\mathbf{A}} {\mathbf{\bar c}} -  \varsigma \tilde\sigma _{\text{n}}^2 \\
		&{\rm s.t.} \ \ \|\mathbf{\bar c}^{(n)}\|_2^2 = 1 , \forall n,
	\end{align}
	\end{subequations}
	where the definition of coefficient matrix $\mathbf{A}$ in this problem is given by
	\begin{equation}
		{\mathbf{A}} = \frac{{\tilde \beta }}{{\sigma _{\text{n}}^2}}{\mathbf{A}}_{{\text{EM,c}}}^{} +  \beta \omega{\mathbf{A}}_{{\text{EM,t}}}^{} - \varsigma \sum\limits_{\kappa \in \mathcal{I}} \omega_{\kappa} {{{\mathbf{A}}_{{\text{EM,}}\kappa}}} .
	\end{equation}

	Different from $\mathcal{P}_{1.2}$ in (\ref{optsustfdbf}), we cannot adopt the Rayleigh-Ritz method here since the feasible set is a Cartesian product of unit spheres (a product manifold). Fortunately, \ac{sdr} can be adopted here. Defining $\mathbf{\bar C} \triangleq \mathbf{\bar c}\mathbf{\bar c}^{\rm H}$, the optimization problem after \ac{sdr} can be reformulated as 
	\begin{subequations}\label{optsustemsdr}
		\begin{align}
			\mathcal{P}_{1.4}:\;&\mathop {\max }\limits_{{\mathbf{\bar C}}}  {\rm Tr}\left(\mathbf{A {\bar C}}\right)-  \varsigma \tilde \sigma _{\text{n}}^2 \\
			&{\rm s.t.} \ \  {\rm Tr}\left(\mathbf{P}_n\mathbf{\bar C}\right) = 1,\; \forall n,\\
			&\;\;\;\;\;\;\;\mathbf{\bar C} \in \mathbb{H}^{+},
		\end{align}
	\end{subequations}
	where $\mathbf{P}_n$ is a block diagonal matrix, with only the $n$th block being the identity matrix $\mathbf{I}_T$. 
	$\varsigma$ is an auxiliary variable introduced by Dinkelbach's transform as before, with its optimal solution iterating as
	\begin{equation}\label{lambda2}
		\varsigma ^ \star  = \beta \omega {\text{Tr}}\left( {{\mathbf{A}}_{{\text{EM,t}}}^{}{{\mathbf{\bar C}}}} \right)/\left( {\sum\limits_{\kappa \in \mathcal{I}} \omega_{\kappa} {{\text{Tr}}\left( {{\mathbf{A}}_{{\text{EM,}}\kappa}^{}{{\mathbf{\bar C}}}} \right) + \tilde\sigma _{\text{n}}^2} } \right).
	\end{equation}

	
	Solving the \ac{sdp} problem suffers from high computational complexity. Nevertheless, we show that it can be significantly reduced by introducing an \ac{ao} across different antennas. For the $n$th antenna, the objective function can be rewritten with respect to $\mathbf{c}^{(n)}$ as shown in (\ref{objfunc_c_sep}),
	\begin{figure*}
		\begin{equation}\label{objfunc_c_sep}
		{f_{{\text{obj}}}}\left( {{{\mathbf{c}}^{(n)}}} \right) \! = \! \underbrace {\frac{{\tilde \beta }}{{\sigma _n^2}}\left( {{{\mathbf{c}}^{(n){\text{H}}}}{\mathbf{A}}_{{\text{EM}},1}^{(n)}{{\mathbf{c}}^{(n)}} \! + \! 2\Re \{ {{\mathbf{c}}^{(n){\text{H}}}}{\mathbf{b}}_{{\text{EM}},1}^{(n)}\}  \! + \! d_{{\text{EM}},1}^{(n)}} \right)}_{{\text{Communication Sum - Rate Part}}} \!+\! \underbrace {\frac{{\beta \omega \left( {{{\mathbf{c}}^{(n){\text{H}}}}{\mathbf{A}}_{{\text{EM}},t}^{(n)}{{\mathbf{c}}^{(n)}} \! + \! 2\Re \{ {{\mathbf{c}}^{(n){\text{H}}}}{\mathbf{b}}_{{\text{EM}},t}^{(n)}\}  \! + \! d_{{\text{EM}},t}^{(n)}} \right)}}{{\sum\limits_{\kappa  \in {\mathcal{I}}}^{{N_{{\text{Scat}}}}} \omega_{\kappa} [ {{{\mathbf{c}}^{(n){\text{H}}}}{\mathbf{A}}_{{\text{EM}},\kappa }^{(n)}{{\mathbf{c}}^{(n)}} \! + \! 2\Re \{ {{\mathbf{c}}^{(n){\text{H}}}}{\mathbf{b}}_{{\text{EM}},\kappa }^{(n)}\}  \! + \! d_{{\text{EM}},\kappa }^{(n)} ]  \! + \! \tilde\sigma _n^2} }}}_{{\text{Sensing SCNR Part}}},{\text{ }}
	\end{equation}
	\begin{equation}\label{DefA}
		{\mathbf{A}^{(n)}} = \frac{{\tilde \beta }}{{\sigma _n^2}}\left[ {\begin{array}{*{20}{c}}
				{{\mathbf{A}}_{{\text{EM}},1}^{(n)}}&{{\mathbf{b}}_{{\text{EM}},1}^{(n)}} \\ 
				{{\mathbf{b}}_{{\text{EM}},1}^{(n)\;{\text{H}}}}&{d_{{\text{EM}},1}^{(n)}} 
		\end{array}} \right] + \beta \omega \left[ {\begin{array}{*{20}{c}}
				{{\mathbf{A}}_{{\text{EM}},t}^{(n)}}&{{\mathbf{b}}_{{\text{EM}},t}^{(n)}} \\ 
				{{\mathbf{b}}_{{\text{EM}},t}^{(n)\;{\text{H}}}}&{d_{{\text{EM}},t}^{(n)}} 
		\end{array}} \right] - \varsigma \sum\limits_{\kappa \in \mathcal{I}}  \omega_{\kappa} {\left[ {\begin{array}{*{20}{c}}
					{{\mathbf{A}}_{{\text{EM}},\kappa}^{(n)}}&{{\mathbf{b}}_{{\text{EM}},\kappa}^{(n)}} \\ 
					{{\mathbf{b}}_{{\text{EM}},\kappa}^{(n)\;{\text{H}}}}&{d_{{\text{EM}},\kappa}^{(n)}} 
			\end{array}} \right]}.
	\end{equation}
	\begin{equation}\label{lambda3}
	\varsigma _{i + 1}^ \star  = \beta \omega {\mathbf{\bar c}}_1^{(n)\;{\text{H}}}\left[ {\begin{array}{*{20}{c}}
			{{\mathbf{A}}_{{\text{EM}},t}^{(n)}}&{{\mathbf{b}}_{{\text{EM}},t}^{(n)}} \\ 
			{{\mathbf{b}}_{{\text{EM}},t}^{(n)\;{\text{H}}}}&{d_{{\text{EM}},t}^{(n)}} 
	\end{array}} \right]{\mathbf{\bar c}}_1^{(n)}/\left( {{\mathbf{\bar c}}_1^{(n)\;{\text{H}}}\sum\limits_{\kappa \in \mathcal{I} } \omega_{\kappa} {\left[ {\begin{array}{*{20}{c}}
					{{\mathbf{A}}_{{\text{EM}},\kappa}^{(n)}}&{{\mathbf{b}}_{{\text{EM}},\kappa}^{(n)}} \\ 
					{{\mathbf{b}}_{{\text{EM}},\kappa}^{(n)\;{\text{H}}}}&{d_{{\text{EM}},\kappa}^{(n)}} 
			\end{array}} \right]} {\mathbf{\bar c}}_1^{(n)} + \tilde\sigma _{\text{n}}^2} \right),
	\end{equation}
	\hrulefill
	\end{figure*}
	where ${{\mathbf{A}}^{(n)}_{{\rm EM},\kappa} } \triangleq {\mathbf{a}}_{{\rm EM}, \kappa} ^{(n)}{({\mathbf{a}}_{{\rm EM}, \kappa}^{(n)})^{\text{H}}}$, ${{\mathbf{b}}^{(n)}_{{\rm EM}, \kappa}} \triangleq {\mathbf{a}}_{{\rm EM}, \kappa} ^{(n)}s_{{\rm EM}, \kappa}^{(n)}$, ${d^{(n)}_{{\rm EM}, \kappa}} \triangleq |{s^{(n)}_{{\rm EM}, \kappa} }{|^2}$, and ${s^{(n)}_{{\rm EM}, \kappa}} \triangleq \sum\limits_{k \ne n} ( {\mathbf{a}}_{{\text{EM}},\kappa }^{(k)}{)^{\text{H}}}{{\mathbf{\bar c}}^{(k)}}$. 
	By adopting Dinkelbach's transform and defining ${\mathbf{\bar c}}_{\rm ext}^{(n)} = {\left[ {\begin{array}{*{20}{c}}
				{{{\mathbf{\bar c}}^{(n)\;{\text{T}}}}}&1 
		\end{array}} \right]^{\text{T}}}$, the optimization problem can be reformulated with respect to $\mathbf{\bar c}_{\rm ext}^{(n)}$ 
	\begin{subequations}\label{optsustemsep}
		\begin{align}
			\mathcal{P}_{1.5} \;&\mathop {\max }\limits_{{\mathbf{\bar c}}_{\rm ext}^{(n)}} {\mathbf{\bar c}}_{\rm ext}^{(n)\;{\text{H}}}{\mathbf{A}^{(n)}   \mathbf{ \bar c}}_{\rm ext}^{(n)} - \varsigma \tilde\sigma _n^2 \hfill \\
			&{\text{s}}{\text{.t}}{\text{.}} \ \ {{\mathbf{e}}^{\text{T}}}{\mathbf{\bar c}}_{\rm ext}^{(n)} = 1,\\
			&\;\;\;\;\;\;\|\mathbf{\bar c}_{\rm ext}^{(n)}\|_2^2 = 2,
		\end{align} 
	\end{subequations}
	where $\mathbf{e} = [0,\;0,\;\cdots,\;1]^{\rm T}$. The definition of $\mathbf{A}^{(n)}$ is given in (\ref{DefA}). $\varsigma$ can be updated as shown in (\ref{lambda3}). To solve this problem, we first give its Lagrangian function as
	\begin{equation}
		\begin{gathered}
			\mathcal{L} = {\mathbf{\bar c}}_{\rm ext}^{(n)\;{\text{H}}}{\mathbf{A}^{(n)}\mathbf{\bar c}}_{\rm ext}^{(n)} - \mu_1 ({\mathbf{\bar c}}_{\rm ext}^{(n)\;{\text{H}}}{\mathbf{\bar c}}_{\rm ext}^{(n)} - 2) - \\
			\mu_2 ({{\mathbf{e}}^{\text{H}}}{\mathbf{\bar c}}_{\rm ext}^{(n)} - 1) - {\mu_2 ^*}({\mathbf{\bar c}}_{\rm ext}^{(n)\;{\text{H}}}{\mathbf{e}} - 1),
		\end{gathered}
	\end{equation}
	where $\mu_1$ and $\mu_2$ are the Lagrangian multipliers for constraints. Solving the Wirtinger gradient for ${\mathbf{\bar c}}_1^{(n)}$ yields 
	\begin{equation}
		\frac{{\partial \mathcal{L}}}{{\partial {\mathbf{\bar c}}_{\rm ext}^{(n)\;*}}} = {\mathbf{A}^{(n)}\mathbf{\bar c}}_{\rm ext}^{(n)} - \mu_1 \mathbf{\bar c}_{\rm ext}^{(n)} - {\mu_2^*}{\mathbf{e}} .
	\end{equation}
	Setting the gradient to zero yields the optimal solution as
	\begin{equation}\label{eq:closedformcn}
		{\mathbf{\bar c}}_{\rm ext}^{(n)\; \star } = \frac{{{{\left( {{{\mathbf{A}}^{(n)}} - {\mu _1}\mathbf{I}} \right)}^{ - 1}}{\mathbf{e}}}}{{{{\mathbf{e}}^{\text{T}}}{{\left( {{{\mathbf{A}}^{(n)}} - {\mu _1}\mathbf{I}} \right)}^{ - 1}}{\mathbf{e}}}}.
	\end{equation}
	
	\begin{proposition}\label{Pro1}
		\normalfont
		Let $\left\{ {{\varpi _1},{\varpi _2}, \cdots ,{\varpi _{T + 1}}} \right\}$ denote the  eigenvalues of $\mathbf{A}^{(n)}$ in an increasing order.
		There exists a unique optimal $\mu_1^{\star}$ over the interval $(\varpi_{T+1},+\infty)$, which can be obtained via binary search.
	\end{proposition}
	\begin{proof}
	See Appendix \ref{AppdixC} for the proof.
	\end{proof}

	\subsubsection{\ac{ra} Model II}
	For Model II, we recall that the structure of the RA channel remains unchanged, with only $T$, $\mathbf{\bar c}^{(n)}$, and $\mathbf{b}$ replaced by $S$, $\mathbf{\tilde c}^{(n)}$, and $\mathbf{\tilde b}$, respectively. As a result, the digital beamforming optimization framework remains the same as that of Model I. Accordingly, the channel covariance matrix in (\ref{RallmodelI}) remains unchanged as
	\begin{equation}
		{\mathbf{R}} = \tilde \beta {{\mathbf{R}}_{{\text{EM,c}}}}/\sigma _{\text{n}}^2 +  \beta \omega {{{\mathbf{R}}_{{\text{EM,t}}}} - \varsigma \sum \limits_{\kappa \in \mathcal{I}} \omega_{\kappa} {{\mathbf{R}}_{{\text{EM,}\kappa}}}},
	\end{equation}
	where ${\mathbf{R}}_{{\text{EM,c}}}$, ${\mathbf{R}}_{{\text{EM,t}}}$, and ${\mathbf{R}}_{{\text{EM,}\kappa}}$ keep the same forms as in Model I equation (\ref{RmodelI}) as
	\begin{equation}
		{{\mathbf{R}}_{{\text{EM},\kappa}}} = {\left( {{\mathbf{h}}_{\kappa}^{{\text{EM}}\;{\text{H}}}{{\mathbf{ F}}_{{\text{EM}}}}} \right)^{\text{H}}}{\mathbf{h}}_{\kappa}^{{\text{EM}}\;{\text{H}}}{{\mathbf{ F}}_{{\text{EM}}}},
	\end{equation}
	except that the \ac{em} beamforming matrix $\mathbf{F}_{\rm EM}$ is now constructed from the discrete radiation-state vector $\tilde{\mathbf c}$ in Model II. Therefore, the auxiliary variable $\varsigma$ and optimal digital beamformer $\mathbf{f}_{\rm FD}$ is updated iteratively as
	\begin{equation}\label{lambda1_II}
		\varsigma ^ \star  = \beta \omega {\mathbf{f}}_{{\text{FD}}}^{{\text{H}}}{{\mathbf{R}}_{{\text{EM,t}}}}{\mathbf{f}}_{{\text{FD}}}/\left( {{\mathbf{f}}_{{\text{FD}}}^{{\text{H}}}\sum\limits_{\kappa \in \mathcal{I}}  \omega_{\kappa}{{{\mathbf{R}}_{{\text{EM,}}{{\kappa }}}}} {\mathbf{f}}_{{\text{FD}}}+ \tilde \sigma _{\text{n}}^2} \right),
	\end{equation}
	
	\begin{equation}\label{fdsolution_II}
		\mathbf{f}_{\rm FD}^{\star} = \sqrt{P} \mathbf{u}.
	\end{equation}

	To tackle the integer constraints in (\ref{selecon}) in a low-complexity manner, we first conduct an \ac{ao} across different antennas. Then, the optimal solution $\mathbf{\tilde c}^{(n)\;\star}$ can be obtained by evaluating the objective in (\ref{objfunc_c}) over all possible $S$ candidates. Because the coefficient vector $\mathbf{\tilde c}^{(n)}$ is 1-sparse (only one nonzero entry), evaluating the metric for a fixed candidate pattern only involves accessing one element from each channel block and accumulating over the ${{N_{{\text{Scat}}}}}$ interference terms. 
	As a result, the per-pattern computation is lightweight, making exhaustive traversal over the $S$ \ac{ra} candidates computationally affordable in this step.
	With Model II, we alternately optimize fully digital and \ac{em} beamforming in the outer layer, while alternately optimizing the \ac{em} beamforming for each antenna in the inner layer, ultimately converging upon a local optimum.
	
	\subsection{Overall Algorithm and Complexity Analysis}
	\begin{algorithm}[t]
		\color{black}
		\caption{Algorithm of solving $\mathcal{P}_{1}$ in (\ref{optsust}).}\label{sustalg}
		\begin{algorithmic}[1]
			\STATE {\bf Initialize} 
			Initial feasible solution: 
			$\mathbf{F}_{\rm RF}$ $\mathbf{f}_{\rm BB}$, $\mathbf{c}^{(n)} $, and $\varsigma$.
			\REPEAT 
			\REPEAT
			\STATE  Calculate the optimal fully digital beamformer according to (\ref{fdsolution});\\
			\STATE Update $\varsigma$ according to (\ref{lambda1}); \\
			\UNTIL  $\mathcal{P}_{1.2}$ in (\ref{optsustfdbf}) converges;\\ 
			\STATE Solve \ac{em} beamforming according to Alg. \ref{embfalg}; \\ 
			\UNTIL  $\mathcal{P}_{1}$ in (\ref{optsust}) converges; \\
			\STATE Decompose the fully digital beamformer into analog and digital one according to \cite{HBFApp1,HBFApp2}. 
			\ENSURE $\mathbf{F}_{\rm EM}^{\star}$, $\mathbf{F}_{\rm RF}^{\star}$, and $\mathbf{f}_{\rm BB}^{\star}$.\\
		\end{algorithmic}
	\end{algorithm}
	
	\begin{algorithm}[t]
		\color{black}
		\caption{Algorithm of Solving $\mathcal{P}_{1.3}$ in (\ref{optsustem}).}\label{embfalg}
		\begin{algorithmic}[1]
			\STATE {\bf Initialize} Index: $n = 1$.  Initial feasible solution:  $\mathbf{\bar c}$ or $\mathbf{\tilde c}$ ;
			\IF{Joint Optimization under Model I}
			\REPEAT
				\STATE Solve $\mathcal{P}_{1.4}$ in (\ref{optsustemsdr}) for $\mathbf{\bar c}^{\star}$; \\
				\STATE Update $\varsigma$ according to (\ref{lambda2}); \\
			\UNTIL $\mathcal{P}_{1.4}$ in (\ref{optsustemsdr}) converges;
			\ELSIF{Alternating Optimization}
			\REPEAT
				\REPEAT 
					\IF{Model I}
					\STATE Use (\ref{eq:closedformcn}) for the optimal EM beamformer $\mathbf{\bar c}^{(n)}$;
					\ELSIF{Model II}
					\STATE Select the optimal EM beamformer $\mathbf{\tilde c}^{(n)}$ to maximize (\ref{objfunc_c});
					\ENDIF
					\STATE Update $\varsigma$ according to (\ref{lambda3}); \\
				\UNTIL  $\mathcal{P}_{1.5}$ in (\ref{optsustemsep}) converges; \\
				\STATE $n \leftarrow 1 + (n\; {\bmod}\; N_{\rm T})$;
			\UNTIL $\mathcal{P}_{1.3}$ in (\ref{optsustem}) converges;\\
			\ENDIF 
			\ENSURE $\mathbf{\bar c}^{\star}$ or $\mathbf{\tilde c}^{\star}$ .\\
		\end{algorithmic}
	\end{algorithm}

	
	In this part, the proposed joint \ac{thbf} optimization algorithm for $\mathcal{P}_{1}$ in (\ref{optsust}) is summarized step by step in Alg.~\ref{sustalg}, while the detailed implementation of the \ac{em} beamformer design is separately presented in Alg.~\ref{embfalg}.
	The computational complexity of $\mathcal{P}_{1.2}$ in (\ref{optsustfdbf}) mainly arises from the eigenvalue decomposition in (\ref{fdsolution}), with complexity order $\mathcal{O}(N_{\rm T}^3)$. The complexity of \ac{em} beamforming under Model I is $\mathcal{O}(((N_{\rm T}T)^2+N_{\rm T})^{3.5})$ for the joint \ac{sdr} design, whereas the proposed antenna-wise \ac{ao} scheme reduces it to $\mathcal{O}(T^3)$ for each antenna. For \ac{ra} Model II, the complexity of exhaustively scanning all $S$ candidate patterns for one antenna is $\mathcal{O}(N_{\rm Scat}S)$. These results are summarized in Table~\ref{tab:complexity_summary}.
	
	\begin{table}[ht]
		\centering
		\caption{Computational complexity under the \ac{sust} scenario.}
		\label{tab:complexity_summary}
		\setlength{\tabcolsep}{4pt}
		\begin{tabular}{l l}
			\toprule
			\textbf{Module} & \textbf{Order} \\
			\midrule
			FD beamforming (\(\mathcal P_{1.2}\)) & \(\mathcal O(N_{\rm T}^3)\) \\
			EM beamforming, Model I (joint \ac{sdr}) & \(\mathcal {O}(((N_{\rm T}T)^2+N_{\rm T})^{3.5})\) \\
			EM beamforming, Model I (\ac{ao}) & \(N_{\rm T}\mathcal O(T^3)\) \\
			RA beamforming, Model II (exhaustive) &   \(N_{\rm T}\mathcal O(N_{\rm Scat}S)\) \\
		\bottomrule
	\end{tabular}
	\end{table}
	
	\section{Extension to \ac{mumt} Scenarios }
	In this section, we extend the \ac{thbf} design of \ac{ra}-\ac{isac} systems from the \ac{sust} scenario to general \ac{mumt} scenarios, based on both \ac{ra} models I and II. Firstly, the optimization problem can be formulated as follows
	\begin{subequations}\label{optmumt}
		\begin{align}
			&\mathcal{P}_{2}:\;\mathop {\max }\limits_{{\mathbf{F}_{\rm EM}},{\mathbf{F}_{\rm RF}},{\mathbf{f}_{\rm BB}}} \tilde{\beta} R_{c} + \beta \sum_{i=1}^{N_{\rm Tar}} \eta_i\\
			&{\rm s.t.} \ \  \left\| {{{\mathbf{F}}_{{\text{RF}}}}{{\mathbf{F}}_{{\text{BB}}}}} \right\|_F^2 \leqslant P, \\
			&\ \ \ \ \  \ {\rm RA \;constraints \; w.r.t. } \mathbf{F}_{\rm EM}.
		\end{align}
	\end{subequations}
	
	\subsection{Objective Function Transformation via \ac{fp}}
	The main challenge to solve this problem is the  non-convex objective function with sum-of-ratios form. According to \cite{FP}, Dinkelbach's method mentioned above is only applicable to single-ratio problems. Therefore, we consider the multi-ratio \ac{fp} framework proposed in \cite{FP}.
	Firstly, using Lagrangian dual transform given in \cite{FP}, we have a more tractable objective  as
	\begin{equation}\label{LagObjFun}
		{f_{{\text{Lag}}}} =
		\sum\limits_{k = 1}^{{N_{{\text{CU}}}}} \tilde \beta \left[ {a_k + \frac{{\left( {1 + {\tilde \gamma _k}} \right){A_k}}}{{{A_k} + {B_k}}}} \right] + \sum\limits_{i = 1}^{{N_{{\text{Tar}}}}} {\beta \frac{{{C_{i}}}}{{{D_i}}}}  ,
	\end{equation}
	where $a_k = {\log \left( {1 + {\tilde \gamma _k}} \right) - {\tilde \gamma _k}}  $ and $\tilde \gamma_k$ is introduced as an auxiliary variable for each ratio. $A_k/B_k$, $A_k, \; B_k, \; C_i, \; D_i $ have already been defined in (\ref{comsinr}) and (\ref{sensinr}). Following the same procedure as in the \ac{sust} scenario, we first combine the analog beamforming and digital beamforming into an unconstrained fully digital beamformer $\mathbf{F}_{\rm FD}$, with its $k$th column being $\mathbf{f}_{{\rm FD},k} = \mathbf{F}_{\rm RF} [\mathbf{F}_{\rm BB}]_{:,k}$. At this point, $C_i$ and $D_i$ can be written as summations over $\mathbf{f}_{{\rm FD},k}$, as shown below
	\begin{equation}
	{C_i} = \sum\limits_{k = 1}^{{N_{{\text{CU}}}}} {{C_{i,k}} = } \omega_i \sum\limits_{k = 1}^{{N_{{\text{CU}}}}} {\left| {{\mathbf{h}}_{{{\text{t}}_i}}^{{\text{EM}}\;{\text{H}}}{{\mathbf{F}}_{{\text{EM}}}}{{\mathbf{f}}_{{\text{FD}},k}}} \right|_2^2},
	\end{equation}
	\begin{equation}
	{D_i} = \sum\limits_{j = 1}^{{N_{{\text{CU}}}}} {\sum\limits_{\kappa  \in \mathcal{I}_i} \omega_{\kappa}{\left| {{\mathbf{h}}_\kappa ^{{\text{EM H}}}{{\mathbf{F}}_{{\text{EM}}}}{{\mathbf{f}}_{{\text{FD}},j}}} \right|_2^2} }  + \tilde\sigma _{\text{n}}^2.
	\end{equation}
	Substituting the above two equations into (\ref{LagObjFun}) and applying the quadratic transform in \cite{FP}, the objective function can be further reformulated as
	\begin{flalign}\label{QuaObjFun}
	&{f_{{\text{qua}}}} \! = \! \sum\limits_{k = 1}^{N_{\rm CU}}  \left\{ {\tilde \beta } a_k \! + \! 2\Re \left( {p_k^*\sqrt {\tilde \beta (1 \! + \! {\tilde\gamma _k}){A_k}} } \right)  \! - \! {\left| {{p_k}} \right|^2}\left( {{A_k} \! + \! {B_k}} \right)  \right. \notag \hfill \\
	& \left. + \sum\limits_{i = 1}^{{N_{{\text{Tar}}}}} \left[{2\Re \left( {q_{i,k}^*\sqrt {\beta {C_{i,k}}} } \right) - {{\left| {{q_{i,k}}} \right|}^2}{D_i}} \right] \right\}, \hfill 
	\end{flalign}
	where $p_k$ and $q_{i,k}$ are auxiliary variables.

	\subsection{Optimizing $\mathbf{F}_{\rm RF}$ and $\mathbf{F}_{\rm BB}$ with given $\mathbf{F}_{\rm EM}$}
	With the decoupled objective function (\ref{QuaObjFun}) at hand, the optimization for \ac{hbf} can be cast as follows
	\begin{subequations}\label{optmumtfd}
		\begin{align}
			\mathcal{P}_{2.1}:\;&\mathop {\max }\limits_{{\mathbf{f}_{{\rm FD},k}},\tilde\gamma_k,p_k,q_{i,k}}\; f_{\rm qua} \left(\mathbf{f}_{{\rm FD},k},\tilde\gamma_k,p_k,q_{i,k}\right)\\
			&{\rm s.t.} \ \ \sum\limits_{k = 1}^K {{\text{Tr}}\left( {{{\mathbf{f}}_{{\text{FD},k}}}{\mathbf{f}}_{{\text{FD},k}}^{\text{H}}} \right) \leqslant P}.
		\end{align}%
	\end{subequations}
	This problem is convex \ac{wrt} each individual optimization variable. Therefore, we adopt an \ac{ao} approach, and the closed-form optimal solution for each variable is given as follows
	\begin{equation}\label{gammaupdate1}
		\tilde\gamma _k^ \star  = \frac{{{{\left| {{\mathbf{h}}_{{\rm c}_k}^{\text{H}}{{\mathbf{f}}_{{\text{FD,}}k}}} \right|}^2}}}{{\sum\limits_{j \ne k} {{{\left| {{\mathbf{h}}_{{\rm c}_k}^{\text{H}}{{\mathbf{f}}_{{\text{FD,}}j}}} \right|}^2} + \sigma _{\text{n}}^2} }},
	\end{equation}
	\begin{equation}\label{pupdate1}
		p_k^ \star  = \frac{{\sqrt {\tilde \beta \left( {1 + {\tilde\gamma _k}} \right)} {\mathbf{h}}_{{\rm c}_k}^{\text{H}}{{\mathbf{f}}_{{\text{FD,}}k}}}}{{\sum\limits_{j = 1}^{N_{\rm CU}} {{{\left| {{\mathbf{h}}_{{\rm c}_k}^{\text{H}}{{\mathbf{f}}_{{\text{FD,}}j}}} \right|}^2} + \sigma _{\text{n}}^2} }},
	\end{equation}
	\begin{equation}\label{qupdate1}
		q_{i,k}^ \star  = \frac{{\sqrt {\beta \omega_i} {\mathbf{h}}_{{{\text{t}}_i}}^{\text{H}}{{\mathbf{f}}_{{\text{FD,}}k}}}}{{\sum\limits_{j = 1}^{{N_{{\text{CU}}}}} {\sum\limits_{\kappa  \in \mathcal{I}_i} \omega_{\kappa} {\left| {{\mathbf{h}}_\kappa ^{\text{H}}{{\mathbf{f}}_{{\text{FD}},j}}} \right|_2^2} }  + \tilde\sigma _{\text{n}}^2}}.
	\end{equation}
	The closed-form solution of the optimal beamformer $\mathbf{f}_{{\rm FD},k}^{\star}$ is given in (\ref{BFOptSol}), where $\mu^{\star}$ is the optimal Lagrangian dual variable of the power budget constraint obtained by the line search. As discussed in the previous section, the transition from fully digital to hybrid beamformer can be readily achieved via the approximation method in \cite{HBFApp1,HBFApp2}.
	
	\begin{figure*}[hb]
		
		\centering
		\rule{\linewidth}{1pt} 
		\begin{equation}\label{BFOptSol}
			{\mathbf{f}}_{{\text{FD,}}k}^ \star  = {\left[ {{\sum\limits_{j = 1}^{{N_{{\text{CU}}}}} \left| {{p_j}} \right|^2} {{{\mathbf{h}}_{{{\text{c}}_j}}}{\mathbf{h}}_{{{\text{c}}_j}}^{\text{H}}}  + \sum\limits_{i = 1}^{{N_{{\text{Tar}}}}} {\left( {\sum\limits_{j = 1}^{{N_{{\text{CU}}}}} {{{\left| {{q_{i,j}}} \right|}^2}} } \right)\sum\limits_{\kappa  \in {\mathcal{I}_i}} \omega_{\kappa} {{{\mathbf{h}}_\kappa }{\mathbf{h}}_\kappa ^{\text{H}}} }  + {\mu ^ \star }{\mathbf{I}}} \right]^\dag }\left( {{p_k}\sqrt {\tilde \beta \left( {1 + {\tilde\gamma _k}} \right)} {{\mathbf{h}}_{{{\text{c}}_k}}} + \sum\limits_{i = 1}^{{N_{{\text{Tar}}}}} {{q_{i,k}}\sqrt{\beta \omega_i}  {{\mathbf{h}}_{{{\text{t}}_i}}}} } \right),
		\end{equation}
		
	\end{figure*}

	\subsection{Optimizing $\mathbf{F}_{\rm EM}$ with given $\mathbf{F}_{\rm RF}$ and $\mathbf{F}_{\rm BB}$}
	\subsubsection{\ac{ra} Model I}
	Following the same matrix transformation as in Sec. \ref{sustEM}, $A_k$, $B_k$, $C_{i,k}$, and $D_i$ can be written \ac{wrt} $\mathbf{\bar c}$ as ${A_k}  = {{{\mathbf{\bar c}}}^{\text{H}}}{{\mathbf{A}}_{{\text{EM,}}k{\text{,}}{{\text{c}}_k}}}{\mathbf{\bar c}}$, 
	${B_k}   = \sum_{j \ne k} {{{{\mathbf{\bar c}}}^{\text{H}}}{{\mathbf{A}}_{{\text{EM,}}j{\text{,}}{{\text{c}}_k}}}{\mathbf{\bar c}}}  + \sigma _{\text{n}}^2$, 
	${C_{i,k}} = \omega_i{{{\mathbf{\bar c}}}^{\text{H}}}{{\mathbf{A}}_{{\text{EM,}}k{\text{,}}{{\text{t}}_i}}}{\mathbf{\bar c}}$, and 
	${D_i}  = \sum_{j = 1}^{{N_{{\text{CU}}}}} {\sum_{\kappa  \in {\mathcal{I}_i}} \omega_{\kappa} {{{{\mathbf{\bar c}}}^{\text{H}}}{{\mathbf{A}}_{{\text{EM,}}j{\text{,}}\kappa }}{\mathbf{\bar c}}} }  + \tilde\sigma _{\text{n}}^2$. Using $\kappa$ and $k$  to represent the channel type and beamformer index,  ${{\mathbf{A}}_{{\text{EM,}}k{\text{,}}{{\text{c}}_k}}} = {\mathbf{a}}_{{\text{EM,}}k{\text{,}}\kappa }^{}{\mathbf{a}}_{{\text{EM,}}k{\text{,}}\kappa }^{\text{H}}$, and ${\mathbf{a}}_{{\text{EM,}}k{\text{,}}\kappa }^{}$ is constructed from ${\mathbf{a}}_{{\text{EM,}} k, \kappa }^{(n)} $ following the arrangement in (\ref{eq:arrangement}), where
	\begin{equation}
	{\mathbf{a}}_{{\text{EM,}} k, \kappa }^{(n)} = {\left[ {{\mathbf{f}}_{{\text{FD,}}k}^*} \right]_n}{\left[ {{\mathbf{h}}_\kappa ^{{\text{EM}}}} \right]_{\left( {n - 1} \right)T + 1:nT}}.
	\end{equation}
	Therefore, the \ac{em} beamforming optimization problem under \ac{ra} model I can be formulated as
		\begin{subequations}\label{optmumtem}
		\begin{align}
			\mathcal{P}_{2.2}:\;&\mathop {\max }\limits_{{\mathbf{\bar c}},\tilde\gamma_k,p_k,q_{i,k}}\; f_{\rm qua} \left( \mathbf{\bar c},\tilde\gamma_k,p_k,q_{i,k}\right)\\
			&{\rm s.t.} \ \  {\|{\mathbf{\bar c}}^{(n)}\|_2^2 = 1},\; \forall n.
		\end{align}%
	\end{subequations}
	It is straightforward to see that the above matrix operations reformulate the optimization variables in terms of $\mathbf{\bar c}$, while preserving the essential properties of the objective function in (\ref{QuaObjFun}). Therefore, the same strategy can still be employed to update $\tilde\gamma_k$, $p_k$, and $q_{i,k}$ as follows
	\begin{equation}
		{\tilde\gamma^{\star}_k} = \frac{{{{{\mathbf{\bar c}}}^{\text{H}}}{{\mathbf{A}}_{{\text{EM,}}k{\text{,}}{{\text{c}}_k}}}{\mathbf{\bar c}}}}{{\sum\limits_{j \ne k} {{{{\mathbf{\bar c}}}^{\text{H}}}{{\mathbf{A}}_{{\text{EM,}}j{\text{,}}{{\text{c}}_k}}}{\mathbf{\bar c}}}  + \sigma _{\text{n}}^2}},
	\end{equation}
	\begin{equation}
		{p^{\star}_k} = \frac{{\sqrt {\tilde \beta \left( {1 + {\tilde\gamma _k}} \right)} {\mathbf{a}}_{{\text{EM,}}k{\text{,}}{{\text{c}}_k}}^{\text{H}}{\mathbf{\bar c}}}}{{\sum\limits_{j = 1}^{{N_{{\text{CU}}}}} {{{{\mathbf{\bar c}}}^{\text{H}}}{{\mathbf{A}}_{{\text{EM,}}j{\text{,}}{{\text{c}}_k}}}{\mathbf{\bar c}}}  + \sigma _{\text{n}}^2}},
	\end{equation}
	\begin{equation}
		{q^{\star}_{i,k}} = \frac{{\sqrt{\beta \omega_i}  {\mathbf{a}}_{{\text{EM,}}k{\text{,}}{{\text{t}}_i}}^{\text{H}}{\mathbf{\bar c}}}}{{\sum\limits_{j = 1}^{{N_{{\text{CU}}}}} {\sum\limits_{\kappa  \in {\mathcal{I}_i}} \omega_{\kappa}{{{{\mathbf{\bar c}}}^{\text{H}}}{{\mathbf{A}}_{{\text{EM,}}j{\text{,}}\kappa }}{\mathbf{\bar c}}} }  + \tilde\sigma _{\text{n}}^2}}.
	\end{equation}
	
	
	Unfortunately, a closed-form solution for $\mathbf{\bar c}$ cannot be obtained because constraint (\ref{optmumtem}b) is in fact a spherical constraint. Therefore, following the same approach as before, we apply \ac{sdr} to $\mathcal{P}_{2.2}$ in (\ref{optmumtem}) with $\mathbf{\bar C}  = \mathbf{\bar c} \mathbf{\bar c}^{\rm H}$. Ignoring the subscript, all quadratic terms $\mathbf{\bar c}^{\rm H} \mathbf{A} \mathbf{\bar c}$ in $A_k,\; B_k\; C_{i,k}$, and $D_{i}$ can be rewritten as linear functions of $\mathbf{\bar C}$ as ${\rm Tr}\left(\mathbf{A} \mathbf{\bar C}\right) $. As a result, the problem of obtaining the optimal $\mathbf{\bar C}$
	is reformulated as
	\begin{subequations}\label{optmumtemsdr}
		\begin{align}
			\mathcal{P}_{2.3}:\;&\mathop {\max }\limits_{{\mathbf{\bar C}}}\; f_{\rm qua} \left( \mathbf{\bar C}\right)\\
			&{\rm s.t.} \ \  {\rm Tr}\left(\mathbf{P}_n\mathbf{\bar C}\right) = 1,\; \forall n,\\
			&\;\;\;\;\;\;\;\mathbf{\bar C} \in \mathbb{H}^{+}.
		\end{align}%
	\end{subequations}
	Since the objective function is more complicated than that of $\mathcal{P}_{1.4}$ in (\ref{optsustemsdr}), the \ac{sdr} of this problem is not guaranteed to be tight. In general, the optimal solution of the relaxed problem may not admit a rank-1 structure. To obtain a feasible rank-1 solution, we adopt Gaussian randomization \cite{SDR} and then perform block-wise power normalization to ensure feasibility.

	
	The idea of employing \ac{ao} across antennas to reduce computational complexity can also be applied here. Following the augmentation adopted in $\mathcal{P}_{1.5}$ in (\ref{optsustemsep}), the objective function can be reformulated as
	\begin{equation}\label{objfunmumtsep}
		{f_{{\text{qua}}}} \! = \! \sum\limits_{k = 1}^{N_{\rm CU}} {\tilde{\beta}}a_k +  ({\mathbf{\bar c}}{_{{\text{ext}}}^{(n)})^{\text{H}}}{{\mathbf{A}}^{(n)}}{\mathbf{\bar c}}_{{\text{ext}}}^{(n)} + 2\Re \left\{ {{{( {{{\mathbf{a}}^{(n)}}} )}^{\text{H}}}{\mathbf{\bar c}}_{{\text{ext}}}^{(n)}} \right\} ,
	\end{equation}
	where
	\begin{equation}
	{{\mathbf{A}}^{(n)}} = \sum\limits_{k = 1}^{N_{\rm CU}} \left\{ { - {{\left| {{p_k}} \right|}^2}{\mathbf{A}}_{{A_k} + {B_k}}^{(n)} + \sum\limits_{i = 1}^{N_{\rm Tar}} { - {{\left| {{q_{i,k}}} \right|}^2}{\mathbf{A}}_{{D_i}}^{(n)}} } \right\},
	\end{equation}
	\begin{equation}
		{{\mathbf{a}}^{(n)}} \!= \!\sum\limits_{k = 1}^{{N_{{\text{CU}}}}} \left\{ p_k^*\sqrt {\tilde \beta (1 \!+\! {\tilde\gamma _k})} {{ {{\mathbf{a}}_{{\text{EM,}}k{\text{,}}{{\text{c}}_k}}^{(n)}} }} \!+ \! \sum\limits_{i = 1}^{{N_{{\text{Tar}}}}} {q_{i,k}^*\sqrt{\beta \omega_i}  }   {{\mathbf{a}}_{{\text{EM,}}k{\text{,}}{{\text{t}}_i}}^{(n)}}  \right\} ,
	\end{equation}
	and
	\begin{equation}
	{\mathbf{A}}_{{A_k} + {B_k}}^{(n)}  = \sum\limits_{j=1}^{N_{\rm CU}}  = 	{\mathbf{a}}_{{\text{EM,}} j, c_k }^{(n)} (	{\mathbf{a}}_{{\text{EM,}} j, c_k }^{(n)})^{\rm H} + \sigma_{\rm n}^2 \mathbf{e}\mathbf{e}^{\rm T},
	\end{equation}
	\begin{equation}
		{\mathbf{A}}_{{D_i}}^{(n)}  = \sum\limits_{j=1}^{N_{\rm CU}} \sum\limits_{\kappa \in  \mathcal{I}_i} \omega_{\kappa} {\mathbf{a}}_{{\text{EM,}} j, \kappa }^{(n)} (	{\mathbf{a}}_{{\text{EM,}} j, \kappa }^{(n)})^{\rm H} + \tilde\sigma_{\rm n}^2 \mathbf{e}\mathbf{e}^{\rm T}.
	\end{equation}
	 It is worth emphasizing that the objective function considered here shares the same properties as that in $\mathcal{P}_{1.5}$ in (\ref{optsustemsep}), while the constraints remain unchanged. Therefore, following the same solution procedure, the optimal solution is given by
	\begin{equation}
		{\mathbf{\bar c}}_{{\text{ext}}}^{(n)\; \star } = {\mathbf{Q}}\left( {\frac{{1 + {{\mathbf{e}}^{\text{T}}}{\mathbf{Q}}{{\mathbf{a}}^{(n)}}}}{{{{\mathbf{e}}^{\text{T}}}{\mathbf{Qe}}}}{\mathbf{e}} - {{\mathbf{a}}^{(n)}}} \right),
	\end{equation}
	where ${\mathbf{Q}} = {\left( {{{\mathbf{A}}^{(n)}} - {\mu _1}{\mathbf{I}}} \right)^{ - 1}}$. Following the approach in Appendix \ref{AppdixC}, $\mu_1$ can be determined via bisection search.
	
	\subsubsection{\ac{ra} Model II}
	As before, for \ac{ra} model II, we only consider antenna-wise alternating integer optimization to reduce the complexity, and the objective function in (\ref{objfunmumtsep}) is used to evaluate the performance of different \ac{ra} patterns.
	
	\subsection{Overall Algorithm and Complexity Analysis}

	Compared with the \ac{sust} scenario in the previous section, the \ac{mumt} scenario only differs in the treatment of the objective function and the number of introduced auxiliary variables. Therefore, the same algorithmic framework can still be adopted, with the update of $\varsigma$ replaced by those of $\gamma_k$, $p_k$, and $q_{i,k}$. The dominant computational complexity of solving the fully digital beamforming problem in the \ac{mumt} scenario arises from the matrix inversion in (\ref{BFOptSol}), with complexity order $\mathcal{O}(N_{\rm T}^3)$ for each $\mu$ search per \ac{cu}. For \ac{em} beamforming under Model I, the complexity is $\mathcal{O}(((N_{\rm T}T)^2+N_{\rm T})^{3.5})$ for joint optimization, while the antenna-wise \ac{ao} scheme reduces it to $\mathcal{O}(T^3)$ for each antenna. For \ac{ra} Model II, the complexity is $\mathcal{O}((N_{\rm Scat}+N_{\rm Tar}-1)S)$ for each antenna. Although the dominant complexity orders under the \ac{sust} and \ac{mumt} scenarios are similar, the latter incurs higher computational cost in constructing intermediate matrices and typically requires more iterations. These results are summarized in Table~\ref{tab:complexity_summary_mumt}.
	
			\begin{table}[h]
		\centering
		\caption{Computational complexity under the \ac{mumt} scenario.}
		\label{tab:complexity_summary_mumt}
		\setlength{\tabcolsep}{4pt}
		\begin{tabular}{l l}
			\toprule
			\textbf{Module} & \textbf{Order} \\
			\midrule
			FD beamforming ($\mathcal{P}_{2.1}$) & \(K\mathcal O(N_{\rm T}^3)\) \\
			EM beamforming, Model I (joint \ac{sdr}) & \(\mathcal O(((N_{\rm T}T)^2+N_{\rm T})^{3.5})\) \\
			EM beamforming, Model I (\ac{ao}) & \(N_{\rm T}\mathcal O(T^3)\) \\
			RA beamforming, Model II (exhaustive)& \(N_{\rm T}\mathcal O((N_{\rm Scat}+N_{\rm Tar}-1)S)\) \\
			\bottomrule
		\end{tabular}
		
	\end{table}

	\section{Simulation Results}
	This section presents the simulation results to evaluate the performance of the proposed tri-HBF scheme in ISAC systems. Unless otherwise specified, the simulation parameters are listed in Table \ref{SimSetup}. The locations of \acp{cu}, targets, and scatterers are randomly generated over 1000 Monte Carlo trials. The \ac{ra} Model II pattern data are obtained from \cite{pingjunAntData}. 
	The \ac{ra}-based \ac{thbf} design proposed in this paper is named according to the \ac{ra} model type and the \ac{em} beamforming approach as  \emph{RA Joint SDR Model I}, \emph{ RA AO Closed-form Model I}, \emph{RA AO Brute-force Model II}.
	 Several typical non-reconfigurable antennas are considered as benchmark schemes, specifically an omnidirectional antenna and a directional antenna with pattern function $\sin \theta$ are adopted, labeled as \emph{OA HBF} and \emph{CosA HBF}, respectively.
	
	\renewcommand{\arraystretch}{1.1}
	\begin{table}[htbp]
		\caption{Simulation parameters setup.}
		\centering
		\begin{threeparttable}
			\begin{tabular}{*{3}{>{\centering\arraybackslash}m{1.0cm}>{\centering\arraybackslash}m{1.5cm}>{\centering\arraybackslash}m{5cm}}} 
				\Xhline{2pt}
				\textbf{Parameter} & \textbf{Value} & \textbf{Description} \\
				\Xhline{1pt}
				$f_c$ & 3 GHz & Carrier frequency \\
				$N_{\rm T}$ &  16  & Number of transmit antenna  \\
				$N_{\rm R}$ &  16  & Number of receive antenna  \\
				$N_{\rm Scat}$ & 2 & Number of scatterers \\
				$L_k$ & $5$   & Number of communication channel path \\
				$T$ & 9 & Truncation length\\
				$S$ & 64 & Number of available radiation patterns \\
				$d $ & 0.005 m &  Element spacing\\
				$\sigma_{\rm n}^2$ & -94 dBm & Noise power \\
				\Xhline{2pt}
			\end{tabular}%
		\end{threeparttable}
		\label{SimSetup}%
	\end{table}%

	Fig. \ref{vsSNR} illustrates the objective value versus transmit power under the \ac{sust} scenario with 2 \ac{rf} chains and the \ac{mumt} scenario with 4 \ac{rf} chains. Thanks to the additional \ac{dofs} enabled by \ac{ra} and the proposed \ac{thbf} designs, the \emph{RA AO Brute-force Model II } scheme achieves the same objective performance with approximately 2 dBm lower transmit power compared to \emph{OA HBF} and \emph{ CosA HBF } schemes. By further exploiting the full flexibility of Model I, the \emph{RA Joint SDR Model I} and \emph{RA AO Closed-form Model I } schemes yield a transmit power reduction of up to about 10 dBm. The superior performance of Model I demonstrates the importance of fully leveraging \ac{em}-domain optimization in \ac{ra} designs. The insights gained from these results may serve as useful guidelines for next-generation \ac{ra} architectures, particularly in terms of improving pattern resolution and controllability to better balance communication and sensing objectives.

	\begin{figure}[!t]
		\centering
		\begin{minipage}{0.24\textwidth}
			\centering
			\includegraphics[scale=0.5]{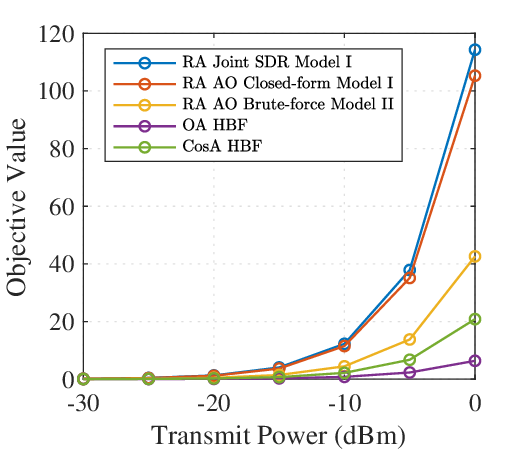}
			\subcaption*{(a)}
			\vspace{-1mm}
		\end{minipage}
		\hfill
		\begin{minipage}{0.24\textwidth}
			\centering
			\includegraphics[scale=0.5]{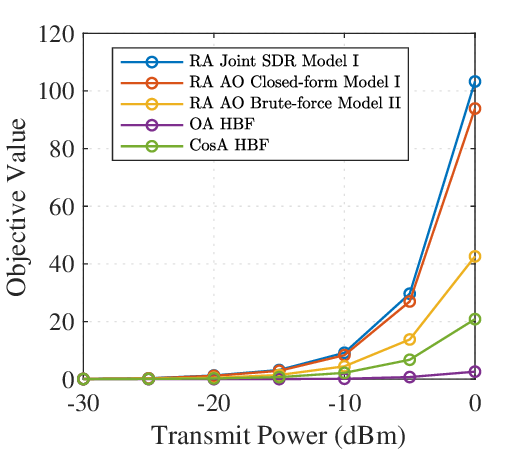}
			\subcaption*{(b)}
			\vspace{-1mm}
		\end{minipage}
		\caption{Objective Value versus transmit power: (a) \ac{sust}; (b)2-\acp{cu} 2-targets scenarios.}\label{vsSNR}
		\vspace{-3mm}
	\end{figure}
	
\begin{figure}[!t]
	\centering
	\begin{minipage}{0.24\textwidth}
		\centering
		\includegraphics[scale=0.5]{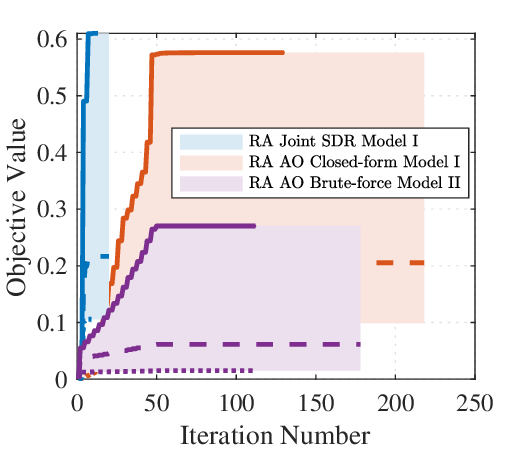}
		\subcaption*{(a)}
		\vspace{-1mm}
	\end{minipage}
	\hfill
	\begin{minipage}{0.24\textwidth}
		\centering
		\includegraphics[scale=0.5]{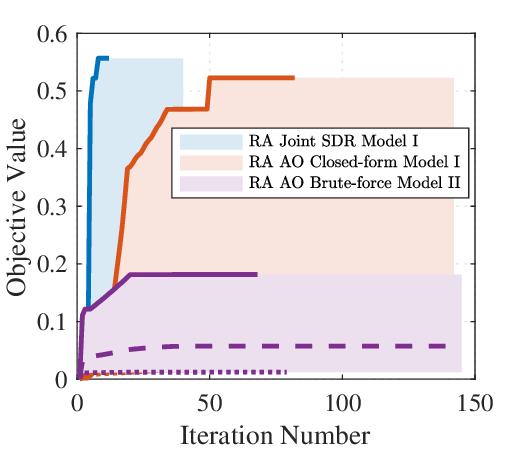}
		\subcaption*{(b)}
		\vspace{-1mm}
	\end{minipage}
	\caption{Convergence behavior : (a) \ac{sust}; (b) 2-\acp{cu} 2-targets scenarios. (Monte Carlo simulation with 1000 trials.)}\label{ConLine}
	\vspace{-3mm}
\end{figure}

\begin{figure}[!h]
	\centering
	\begin{minipage}{0.24\textwidth}
		\centering
		\includegraphics[scale=0.5]{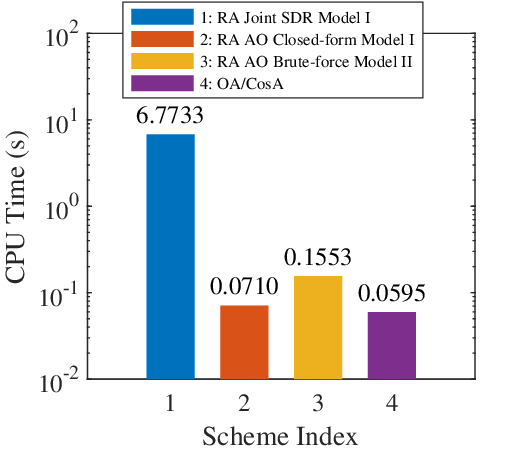}
		\subcaption*{(a)}
		\vspace{-1mm}
	\end{minipage}
	\hfill
	\begin{minipage}{0.24\textwidth}
		\centering
		\includegraphics[scale=0.5]{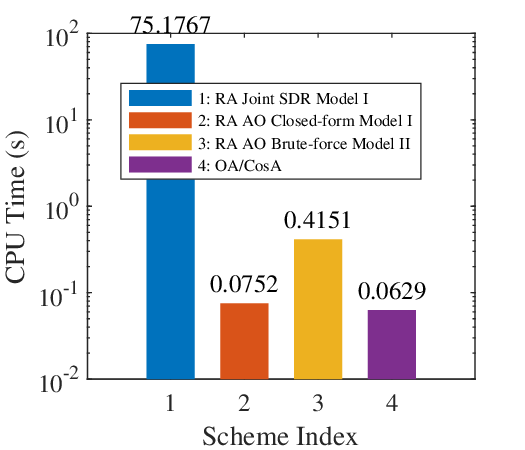}
		\subcaption*{(b)}
		\vspace{-1mm}
	\end{minipage}
	\caption{CPU time comparisons : (a) \ac{sust}; (b) 2-\acp{cu} 2-targets scenarios.}\label{CPUTime}
	\vspace{-3mm}
\end{figure}

\begin{figure}[!h]
	\centering
	\begin{minipage}{0.24\textwidth}
		\centering
		\includegraphics[scale=0.5]{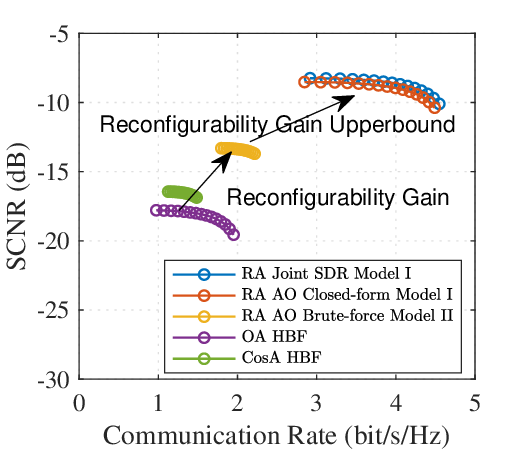}
		\subcaption*{(a)}
		\vspace{-1mm}
	\end{minipage}
	\hfill
	\begin{minipage}{0.24\textwidth}
		\centering
		\includegraphics[scale=0.5]{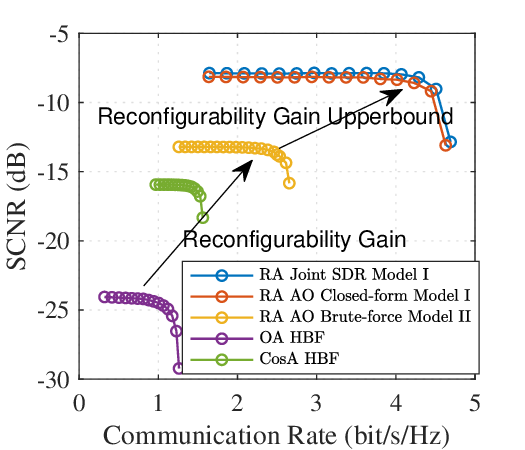}
		\subcaption*{(b)}
		\vspace{-1mm}
	\end{minipage}
	\caption{Trade-off relationship: (a) \ac{sust}; (b) 2-\acp{cu} 2-targets scenarios.}\label{tradeoffline}
	\vspace{-3mm}
\end{figure}

Then, we present the convergence behavior and \ac{cpu} running time of the three proposed schemes in Figs. \ref{ConLine} and \ref{CPUTime}. The transmit power is set to -30 dBm, and the weighting factor $\tilde{\beta}$ is set to $5 \times 10^{-3}$. All simulations are implemented in MATLAB R2024a on a laptop equipped with a 4.05 GHz Apple M3 Pro CPU and 18 GB RAM. 
In Fig. \ref{ConLine}, the solid, dashed, and dotted lines represent the best, average, and worst converged objective values over the 1000 trials, respectively. As observed, under both the \ac{sust} scenario with 2 \ac{rf} chains and the \ac{mumt} scenario with 4 \ac{rf} chains, the \emph{RA Joint SDR Model I} scheme achieves the best performance, followed by the \emph{ RA AO Closed-form Model I} scheme, while the \emph{RA AO Brute-force Model II} scheme performs the worst among the three.
From Fig. \ref{CPUTime}, it can be seen that the \ac{cpu} time exhibits an opposite trend. Under both scenarios, the \emph{RA Joint SDR Model I } scheme requires the longest running time due to the incorporation of \ac{sdr}. The \emph{RA AO Brute-force Model II} scheme ranks next, as it performs exhaustive search over 64 predefined patterns. Benefiting from antenna-wise \ac{ao} with closed-form updates, the \emph{RA AO Closed-form Model I} scheme achieves a significantly reduced computational cost. Compared with the above \ac{ra}-based \ac{thbf} schemes, the \emph{OA/CosA HBF} schemes incur the shortest running time, since no \ac{em} beamforming optimization is required.

Fig. \ref{tradeoffline} illustrates the communication–sensing trade-off performance of different schemes under the \ac{sust} scenario with 2 \ac{rf} chains and the \ac{mumt} scenario with 4 \ac{rf} chains. The transmit power is set to -30 dBm, and the weighting factor $\tilde{\beta}$ is set to $5 \times 10^{-3}$. It can be observed that, in both scenarios, the \emph{RA AO Brute-force Model II} scheme achieves a superior trade-off performance compared to the conventional non-\ac{ra} schemes, namely \emph{OA HBF} and \emph{CosA HBF}. We refer to this performance improvement as the reconfigurability gain.
Building upon this gain, both the \emph{RA Joint SDR Model I} and \emph{RA AO Closed-form Model I} schemes further strike a better trade-off relationship, providing an upper bound on the reconfigurability gain in terms of trade-off performance and flexibility.

\begin{figure*}[!t]
	\centering
	\begin{minipage}[t]{0.24\textwidth}
		\centering
		\includegraphics[width=\linewidth]{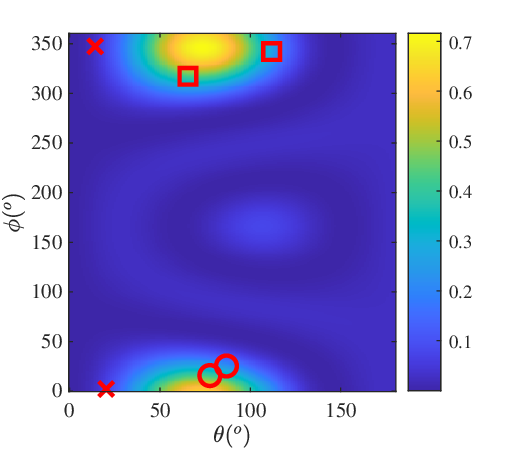}
		\subcaption*{(a) Single \ac{ra} Model I}
		\vspace{-1mm}
	\end{minipage}\hfill
	\begin{minipage}[t]{0.24\textwidth}
		\centering
		\includegraphics[width=\linewidth]{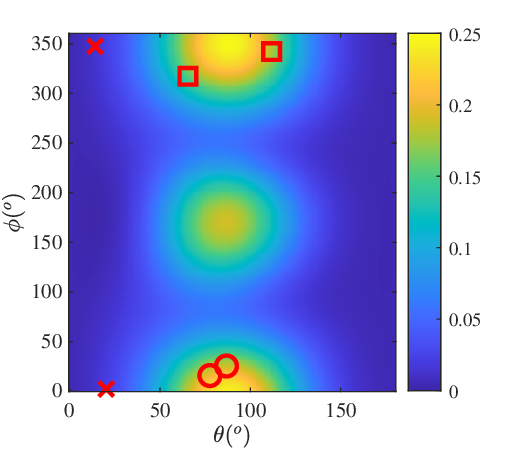}
		\subcaption*{(b) Single \ac{ra} Model II}
		\vspace{-1mm}
	\end{minipage}\hfill
	\begin{minipage}[t]{0.24\textwidth}
		\centering
		\includegraphics[width=\linewidth]{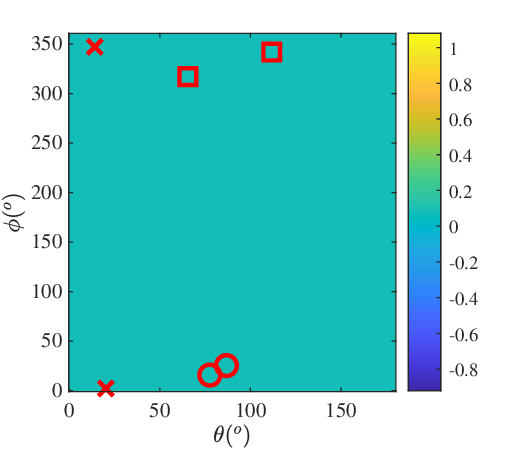}
		\subcaption*{(c) Single OA}
		\vspace{-1mm}
	\end{minipage}\hfill
	\begin{minipage}[t]{0.24\textwidth}
		\centering
		\includegraphics[width=\linewidth]{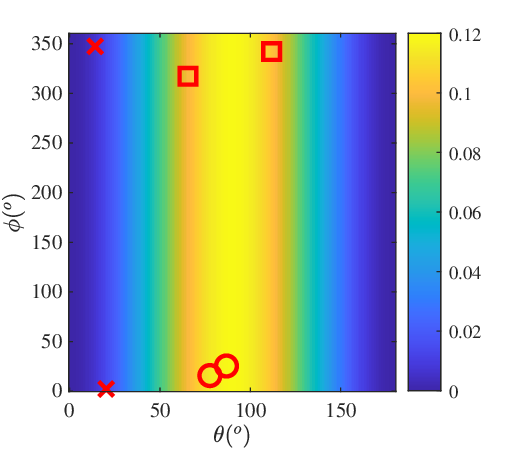}
		\subcaption*{(d) Single CosA}
		\vspace{-1mm}
	\end{minipage}
	
	\vspace{2mm} 
	
	\begin{minipage}[t]{0.24\textwidth}
		\centering
		\includegraphics[width=\linewidth]{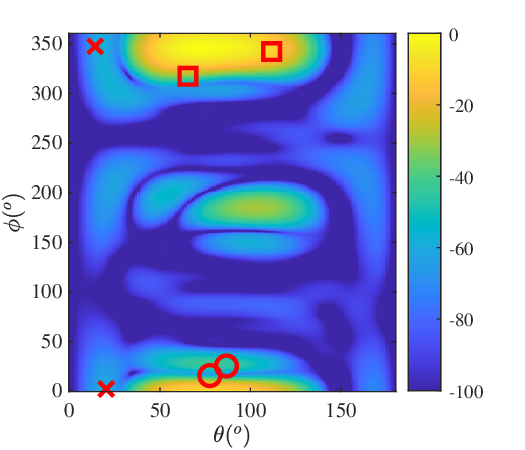}
		\subcaption*{(e) \ac{ra} array Model I}
		\vspace{-1mm}
	\end{minipage}\hfill
	\begin{minipage}[t]{0.24\textwidth}
		\centering
		\includegraphics[width=\linewidth]{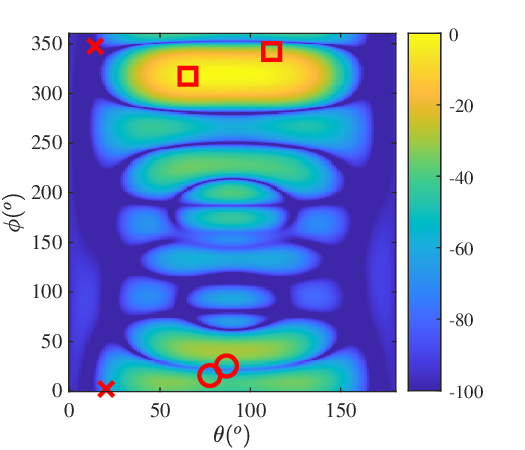}
		\subcaption*{(f) \ac{ra} array Model II}
		\vspace{-1mm}
	\end{minipage}\hfill
	\begin{minipage}[t]{0.24\textwidth}
		\centering
		\includegraphics[width=\linewidth]{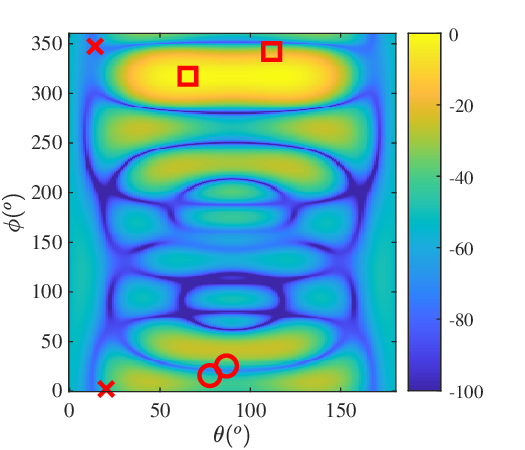}
		\subcaption*{(g) \ac{oa} array}
		\vspace{-1mm}
	\end{minipage}\hfill
	\begin{minipage}[t]{0.24\textwidth}
	\centering
	\includegraphics[width=\linewidth]{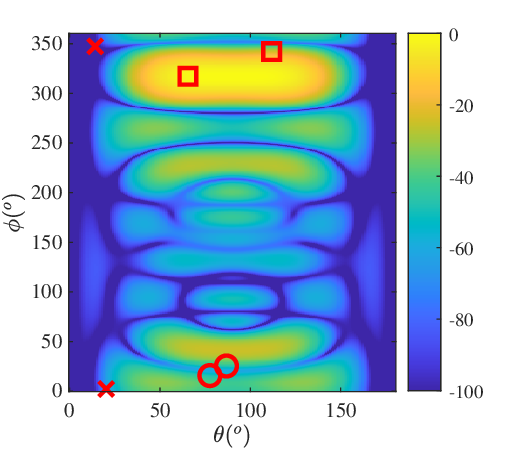}
	\subcaption*{(h) CosA array}
	\vspace{-1mm}
\end{minipage}
	\vspace{0mm}
	\caption{Radiation pattern comparisons where $N_{\rm CU} = 2$, $N_{\rm Tar} = 2$, and $N_{\rm RF} = 4$. (circle: CU, square: target, cross: scatterers).}
	\label{Beampattern}
	\vspace{-2mm}
\end{figure*}

To more clearly illustrate the advantage of RA, the single-antenna radiation patterns and the array transmit beampatterns of different schemes are shown in Fig. \ref{Beampattern}. The transmit power is set to -30 dBm, and the weighting factor $\tilde{\beta}$ is set to $5 \times 10^{-3}$. From the single-antenna perspective, both OA and CosA exhibit fixed radiation coverage. The OA radiates energy uniformly in all directions, whereas the CosA concentrates its radiation around $\theta = 90^{\rm o}$. In contrast, after \ac{thbf} optimization, the \ac{ra} is capable of concentrating radiated energy toward the intended \acp{cu} and sensing targets. Furthermore, owing to the more flexible spherical harmonics representation adopted in Model I, it achieves stronger directivity than Model II, which manifested by reduced power leakage toward other regions and a higher main-lobe peak gain. Intuitively, better single-antenna directivity leads to improved array-level transmit performance, which is validated by the transmit beampatterns shown in Fig. \ref{Beampattern}. Specifically, CosA outperforms OA in suppressing energy at scatterer locations. Additionally, the array beampatterns corresponding to the two \ac{ra} models show less energy leakage, with Model I demonstrating particularly reduced sidelobe levels, thus enhancing power concentration in the main beam.

From a hardware reduction perspective, the objective values versus the number of RF chains and antennas are shown in Figs. \ref{vsNRF} and \ref{vsNT}, where the transmit power is set to 0 dBm and the communication–sensing trade-off factor is fixed at 0.5.
In general, all schemes exhibit improved performance with an increasing number of RF chains and antennas, with the \ac{ra}-based 
\ac{thbf} consistently outperforming the conventional OA and CosA \ac{hbf} schemes. Notably, as shown in Fig. \ref{vsNRF}, the \emph{RA AO Brute-force Model II} scheme, despite its relatively lower performance, outperforms the conventional fully digital beamforming schemes (OA and CosA). Furthermore, compared to \emph{CosA HBF} scheme, \emph{RA AO Brute-force Model II} scheme achieves a reduction of about one RF chain. Thanks to the arbitrary pattern generation ability of RA Model I, both \emph{RA Joint SDR Model I} and \emph{RA AO Closed-form Model I} achieve significantly better performance gains. With only two RF chains, they outperform the \emph{OA HBF} and \emph{CosA HBF} schemes by a considerable margin, showcasing the future opportunities for advanced \ac{ra} design and development.
As illustrated in Fig. \ref{vsNT}, for the same number of RF chains, a limited number of \acp{ra} can replicate the performance of large antenna arrays in conventional OA and CosA schemes. Specifically, Model II can save around 10 antennas compared to the CosA scheme, while Model I results in a saving of approximately 30 antennas. This clearly highlights the advantages of \ac{ra} in environments with limited space for large antenna array.

\begin{figure}
	\centering
	\includegraphics[scale=0.8]{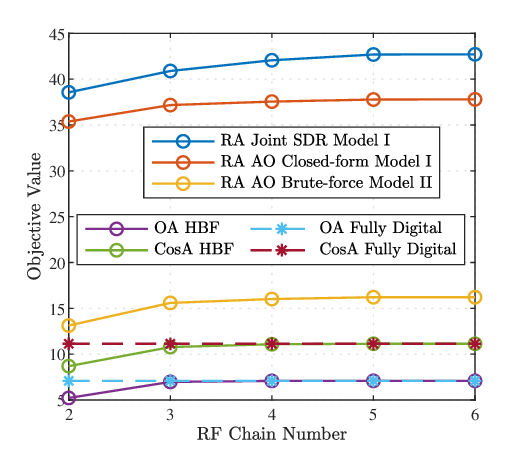}
	\vspace{-1mm}
	\caption{ Objective value versus $N_{\rm RF}$ where $N_{\rm CU} = 2$ and $N_{\rm Tar} = 2$.}\label{vsNRF}
	\vspace{-3mm}
\end{figure}

\begin{figure}
		\centering
	\includegraphics[scale=0.8]{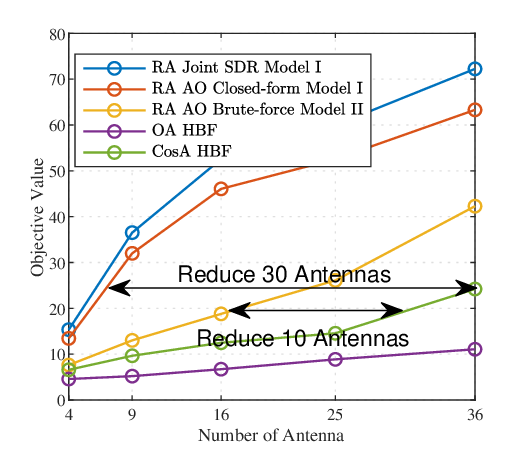}
	\vspace{-1mm}
	\caption{Objective value versus $N_{\rm T}$ where $N_{\rm RF} = 4$, $N_{\rm CU} = 2$ and $N_{\rm Tar} = 2$.}\label{vsNT}
	\vspace{-3mm}
\end{figure}

	\section{Conclusion}
	This paper presents a novel \ac{thbf} design for ISAC system with \ac{ra}. Two \ac{ra} models,  \emph{arbitrary pattern generation} model and \emph{discrete pattern selection} model, are introduced and optimized in both \ac{sust} and \ac{mumt} scenarios. Compared with the conventional hybrid and fully digital schemes with non-\ac{ra}, the proposed schemes achieve superior performance in terms of objective value, communication and sensing trade-off, beam pattern focusing, and hardware overhead reduction, demonstrating the strong potential of \ac{ra} in future \ac{isac} systems.

	\appendices
	\section{}\label{AppdixA}
	Employing \ac{sdr} as $\mathbf{R}_{\rm FD} = \mathbf{f}_{\rm FD}\mathbf{f}_{\rm FD}^{\rm H}
	$, the objective function becomes
	\begin{equation}\label{sdrobj}
		f = \tilde \beta \frac{{{\text{Tr}}\left( {{{\mathbf{R}}_{{\text{EM,1}}}}{{\mathbf{R}}_{{\text{FD}}}}} \right)}}{{\sigma _{\text{n}}^2}} + \beta \frac{{{\text{Tr}}\left( {{{\mathbf{R}}_{{\text{EM,t}}}}{{\mathbf{R}}_{{\text{FD}}}}} \right)}}{{{\text{Tr}}\left( {{{\mathbf{R}}_{{\text{EM,s}}}}{{\mathbf{R}}_{{\text{FD}}}}} \right) + \sigma _{\text{n}}^2}}.
	\end{equation}
	Since both the numerator and the denominator are linear functions with respect to $\mathbf{R}_{\rm FD}$, they are simultaneously convex and concave. As a result, Dinkelbach's transform can be directly employed to achieve the following \ac{sdp}
	\begin{subequations}
		\begin{align}
			\mathop {\max }\limits_{{{\mathbf{R}}_{{\text{FD}}}}} &\frac{{\tilde \beta {\text{Tr}}\left( {{{\mathbf{R}}_{{\text{EM,1}}}}{{\mathbf{R}}_{{\text{FD}}}}} \right)}}{{\sigma _{\text{n}}^2}} + \beta {\text{Tr}}\left( {{{\mathbf{R}}_{{\text{EM,t}}}}{{\mathbf{R}}_{{\text{FD}}}}} \right) \notag \\
			&- \lambda \left[ {{\text{Tr}}\left( {{{\mathbf{R}}_{{\text{EM,s}}}}{{\mathbf{R}}_{{\text{FD}}}}} \right) + \sigma _{\text{n}}^2} \right]  \\
			&{\rm s.t.} \ \ {\rm Tr}\left(\mathbf{R}_{\rm FD}\right) \leq P,\\
			&\;\;\;\;\;\;\;\mathbf{R}_{\rm FD} \in \mathbb{H}^{+}.
		\end{align}
	\end{subequations} 
	
	According to Proposition 3.5 of \cite{SDPRank}, this relaxation is tight since there exists only one constraint. Since a rank-1 solution is always available, we may bypass \ac{sdr} and perform equivalent Dinkelbach's transform directly on the original objective function. This completes the proof.
	
	\section{}\label{AppdixB}
	Following the approach in Appendix \ref{AppdixA}, it is sufficient to show that Problem( \ref{optsustemsdr}) always admits a rank-1 optimal solution. To prove this, the \ac{kkt} conditions of Problem (\ref{optsustemsdr}) is firstly given by
	\begin{equation}
		\left\{ \begin{gathered}
			{{\mathbf{Z}}^ \star } = \mathbf{T}_1 + \mathbf{T}_2 \hfill \\
			{{\mathbf{C}}^ \star } \succcurlyeq {\mathbf{0}},{{\mathbf{Z}}^ \star } \succcurlyeq {\mathbf{0}},\mu _n^ \star  \geqslant 0, \hfill \\
			{\text{Tr}}\left( {{{\mathbf{Z}}^ \star }{{\mathbf{C}}^ \star }} \right) = 0 \Rightarrow {{\mathbf{Z}}^ \star }{{\mathbf{C}}^ \star } = {\mathbf{0}}. \hfill \\ 
		\end{gathered}  \right.
	\end{equation}
	where $\mu_n$ is the Lagrangian multiplier for the constraints (\ref{optsustemsdr}b), $\mathbf{Z}$ denotes the Lagrangian multiplier matrix for constraint (\ref{optsustemsdr}c). $\mathbf{T}_1$ is a rank-$r$ matrix expressed as
	\begin{equation}
		\mathbf{T}_1 \!=\! { (\beta\omega{\mathbf{a}}_{{\text{EM,t}}}^{}{\mathbf{a}}_{{\text{EM,t}}}^{\text{H}} \!-\! \varsigma   \sum\limits_{\kappa \in \mathcal{I}} \! \omega_{\kappa} {{\mathbf{a}}_{{\text{EM,}\kappa}}^m{\mathbf{a}}_{{\text{EM},\kappa}}^{m\;{\text{H}}})}  \!-\! \sum\limits_n \! {\mu _n^ \star {{\mathbf{P}}_n}} }.
	\end{equation}
	$\mathbf{T}_2$ is a rank-1 matrix, given by
	\begin{equation}
		\mathbf{T}_2 = {{{\tilde \beta }}{\mathbf{a}}_{{\text{EM,c}}}^{}{\mathbf{a}}_{{\text{EM,c}}}^{\text{H}}}/{\sigma _{\text{n}}^2}.
	\end{equation}
	The \ac{kkt} structure and matrix rank properties of this problem are the same as those in \cite{rank1proof}. Using the Proof of Proposition 4.1 in \cite{rank1proof}, we can complete the proof.
	
	\section{Proof for Proposition \ref{Pro1}}\label{AppdixC}
		Let $(\mathbf{c}^{ \star } , \mu_1^\star, \mu_2^\star)$ be a stationary point satisfying the \ac{kkt} conditions.
		The Hessian matrix of $\mathcal{L}$ \ac{wrt} $\mathbf{\bar c}_{\rm ext}^{(n)}$ as 
		\begin{equation}
			\nabla_{\mathbf c}^2 \mathcal L
			= 2 \bigl( \mathbf A^{(n)} - \mu_1 \mathbf I \bigr).
		\end{equation}
		Let $\delta \mathbf c$ denote any first-order feasible perturbation satisfying $\mathbf e^{T} \delta \mathbf c = 0$ and $\mathrm{Re} \bigl\{ \mathbf c^{\star H} \delta \mathbf c \bigr\} = 0$, which characterizes the tangent space of the constraint manifold	at $\mathbf c^\star$. Since $\mathbf c^\star$ corresponds to a local maximizer, the second-order necessary condition requires
		\begin{equation}
			\delta \mathbf c^{\rm H}
			\bigl( \mathbf A^{(n)} - \mu_1^\star \mathbf I \bigr)
			\delta \mathbf c
			\le 0,
			\quad \forall\, \delta \mathbf c.
			\label{eq:second_order_condition}
		\end{equation}
		
		Assume by contradiction that
		$\mu_1^\star < {\varpi _{T + 1}}$.
		Then the matrix $\mathbf A^{(n)} - \mu_1^\star \mathbf I$ admits at least one
		positive eigenvalue, implying the existence of a feasible direction
		$\delta \mathbf c$ such that
		\begin{equation}
			\delta \mathbf c^{\rm H}
			\bigl( \mathbf A^{(n)} - \mu_1^\star \mathbf I \bigr)
			\delta \mathbf c > 0,
		\end{equation}
		which contradicts~\eqref{eq:second_order_condition}.
		Hence, $\mu_1^\star \ge {\varpi _{T + 1}}$.
		Moreover, the closed-form solution (\ref{eq:closedformcn}) requires $\mathbf A^{(n)} - \mu_1^\star \mathbf I$ to be nonsingular, which implies that $\mu_1^\star > {\varpi _{T + 1}}$. This completes the proof of the existence of  $\mu_1^{\star}$ over $(\varpi_{T+1},+\infty)$.
		
		Next, we give the proof of monotonicity of $\|\mathbf{c}(\mu_1)\|_2^2$ over $\mu_1 \in (\varpi_{T+1},+\infty)$. Let $\mathbf{Q}(\mu_1) \triangleq {{{\left( {{{\mathbf{A}}^{(n)}} - {\mu _1}\mathbf{I}} \right)}^{ - 1}}}$, we have 
		\begin{equation}
			\|\mathbf{c}(\mu_1)\|_2^2
			=\frac{\mathbf e^{\rm H}\mathbf{Q}^2(\mu_1)\mathbf e}{\left[\mathbf e^{\rm H}\mathbf Q(\mu_1)\mathbf e\right]^2}
			\triangleq \frac{a(\mu_1)}{s(\mu_1)^2}.
		\end{equation}
		Taking the derivative of the above equation yields
		\begin{equation}\label{eq:normcderivative}
			\frac{d}{d\mu_1}\,\|\mathbf{c}(\mu_1)\|_2^2
			=
			\frac{2\bigl(\mathbf{e}^{\rm H}\mathbf{Q}^3\mathbf{e}\bigr)\bigl(\mathbf{e}^{\rm H}\mathbf{Q}\mathbf{e}\bigr)
				-2\bigl(\mathbf{e}^{\mathrm H}\mathbf{Q}^2\mathbf{e}\bigr)^2}{\bigl(\mathbf{e}^{\rm H}\mathbf{Q}\mathbf{e}\bigr)^3}.
		\end{equation}
		Since $\mathbf{Q}$ is a semi-negative definite matrix over $\mu_1^{\star}$ over $(\varpi_{T+1},+\infty)$, the derivative in (\ref{eq:normcderivative}) is non-positive. Therefore, $\|\mathbf{c}(\mu_1)\|_2^2$ is non-increasing \ac{wrt} $\mu_1$. This completes the proof of the monotonicity of the power constraint with respect to $\mu_1$.
		
		Combining the proofs of existence and monotonicity, we conclude that there exists a unique optimal solution for $\mu_1$ within the interval $(\varpi_{T+1},+\infty)$.
	
	\bibliographystyle{IEEEtran}      
	\bibliography{IEEEabrv,references}
\end{document}